\documentclass{article}

\usepackage{amssymb,amsmath,amsthm}
\usepackage{natbib}
\usepackage[left=1.25in, right=1.25in, top=1.25in, bottom=1.25in]{geometry}
\usepackage{makecell}
\usepackage{bm}
\usepackage{graphicx}
\usepackage{xcolor}
\usepackage{enumerate}
\usepackage{setspace}
\newcommand{\E}{\mathbb{E}}
\newcommand{\Var}{\mathrm{Var}}
\newcommand{\Cov}{\mathrm{Cov}}
\newcommand{\Cor}{\mathrm{Cor}}

\newcommand{\R}{\mathbb{R}}
\newcommand{\W}{\mathcal{W}}

\renewcommand{\P}{\mathbb{P}}
\newcommand{\td}{\tilde}

\newcommand{\lb}{\left(}
\newcommand{\rb}{\right)}
\newcommand{\eps}{\epsilon}

\renewcommand{\td}{\tilde}

\newcommand{\RY}{R^{{\scriptscriptstyle Y}}}
\newcommand{\RW}{R^{{\scriptscriptstyle W}}}

\newcommand{\nox}{{\scriptscriptstyle\mathrm{nox}}}
\newcommand{\inter}{{\scriptscriptstyle\mathrm{int}}}
\newcommand{\sls}{{\scriptscriptstyle\mathrm{2SLS}}}

\newcommand{\did}{{\scriptscriptstyle\mathrm{DiD}}}
\newcommand{\mdid}{{\scriptscriptstyle\mathrm{mDiD}}}
\newcommand{\rdd}{{\scriptscriptstyle\mathrm{RDD}}}
\newcommand{\mrdd}{{\scriptscriptstyle\mathrm{mRDD}}}
\newcommand{\indep}{\protect\mathpalette{\protect\independenT}{\perp}}
\newcommand{\ta}{\texttt{a}}
\newcommand{\tn}{\texttt{n}}
\newcommand{\tc}{\texttt{c}}
\newcommand{\rate}{\tau_{\mathrm{r}}}

\def\independenT#1#2{\mathrel{\rlap{$#1#2$}\mkern2mu{#1#2}}}

\DeclareMathOperator*{\argmin}{argmin}

\newtheorem{theorem}{Theorem}[section]
\newtheorem{lemma}{Lemma}[section]

\newtheorem{assumption}{Assumption}[section]

\newtheorem{example}{Example}[section]

\newcommand\blfootnote[1]{%
  \begingroup
  \renewcommand\thefootnote{}\footnote{#1}%
  \addtocounter{footnote}{-1}%
  \endgroup
}

\begin{document}
\title{Causal Interpretation of Regressions With Ranks}
\author{Lihua Lei \thanks{Graduate School of Business and Department of Statistics (by courtesy), Stanford University}}
\maketitle

\blfootnote{\hspace{-0.55cm}The author would like to thank Raj Chetty for the helpful discussion.}

\begin{abstract}
In studies of educational production functions or intergenerational mobility, it is common to transform the key variables into percentile ranks. Yet, it remains unclear what the regression coefficient estimates with ranks of the outcome or the treatment. In this paper, we derive effective causal estimands for a broad class of commonly-used regression methods, including the ordinary least squares (OLS), two-stage least squares (2SLS), difference-in-differences (DiD), and regression discontinuity designs (RDD). Specifically, we introduce a novel primitive causal estimand, the Rank Average Treatment Effect (rank-ATE), and prove that it serves as the building block of the effective estimands of all the aforementioned econometrics methods. For 2SLS, DiD, and RDD, we show that direct applications to outcome ranks identify parameters that are difficult to interpret. To address this issue, we develop alternative methods to identify more interpretable causal parameters. 
\end{abstract}

\section{Introduction}
Rank transformations are often applied to variables without natural units, such as test scores \citep[e.g.][]{krueger1999experimental}), to enhance cross-outcome comparability, or variables with heavy tails, such as income \citep[e.g.][]{chetty2014land,deutscher2023measuring, mogstad2024inference}, to avoid excessive influences of extreme observations. In both cases, it was found that the rank transformation tends to reduce sensitivity to specifications \citep[e.g.][]{dahl2008association}. A recent survey by \cite{chetverikov2023inference} found 119 papers that contain rank regressions in the top Economics journals published since 2013.

Despite the popularity of rank transformations in applied research, the causal interpretation of regression coefficients remains unclear when dependent or independent variables are transformed into ranks. For example, \cite{chetty2020race} obtained an estimate of the relative mobility, defined as the coefficient of the parents' income rank in the regression of the child's income rank, of $0.35$ using the full US sample that pools all races. This was interpreted as ``a 10 percentile increase in parents’ rank is associated with a 3.5 percentile increase in children’s rank on average''. It is known that this coefficient is equivalent to the Spearman's rank correlation coefficient \citep{spearman1904proof} between the parents' and child's incomes, assuming that the income distribution is continuous  \citep[e.g.][]{chetverikov2023inference}. However, it is unclear if the Spearman's correlation measures an interventional effect in any thought experiment where parents' incomes are randomized or exposed to exogeneous shocks. Ideally, we would want to interpret the coefficient through the distributions of the child's counterfactual incomes had the parents' income been a certain amount. 

For regressions with ranks only in one side, even the effective estimands have yet been understood. For example, \cite{krueger1999experimental} regressed the percentile ranks of test scores on the assignment of class types (i.e., a small class, a regular class, or a regular class with teacher aid) along with other controls using the Tennessee STAR data \citep{boyd2007project}. Through the OLS estimate he found that ``the gap in average performance is about 5 percentile points in kindergarten, 8.6 points in first grade, and 5–6 points in second and third grade.'' While this suggests a positive effect of smaller class sizes, it is harder to interpret the effect size. Ideally, we would want to express this coefficient through the distributions of potential test scores had the student been assigned into each class type. The interpretation can be further complicated for more sophisticated econometric methods such as the two-stage least squares, which is also applied to analyze the effect of small classes on STAR data to address non-compliance \citep{krueger1999experimental}.

In this paper, we derive the effective estimands as well as their causal interpretations for a suite of econometric methods including the ordinary least squares (OLS), two-stage least squares (2SLS), difference-in-differences (DiD), and regression discontinuity designs (RDD) under a variety of standard assumptions in the econometrics literature. Unlike \cite{chetverikov2023inference}, we do not discuss statistical inference for these estimators and many of their results can be directly applied here. Moreover, we focus on continuous outcomes and do not study the effect of tie-breaking for discrete outcomes as in \cite{chetverikov2023inference} and \cite{chetverikov2024csranks}. Our paper also differs from the literature that introduces inference methods for ranks of true parameters from multiple heterogeneous units or populations \citep[e.g.][]{bazylik2021finite,mogstad2024inference,chetverikov2024csranks} in that we focus on the causal interpretations of common econometric methods with sample ranks under the superpopulation framework. 

We start by introducing a primitive causal estimand which we call \emph{the Rank Average Treatment Effect (rank-ATE)}. The estimand had been studied in the statistics and machine learning literature, though in very different contexts. Akin to the standard ATE, it compares the distributions of outcomes under two treatment conditions, though not through the mean difference. We prove that the rank-ATE shares many desirable properties with the ATE and they have the same signs in a variety of scenarios.

For a randomized experiment with a binary treatment, we prove that the OLS with outcomes ranked based on the entire sample identifies the rank-ATE. This result remains if extra covariates that are not affected by the treatment are included. More surprisingly, we show that the effective estimand remains the same if the outcomes are ranked based solely on the treatment (or control) units. We call this property \emph{reference robustness} and show that it is only satisfied by the OLS but not other aforementioned  econometric methods. When the assignment is confounded under the selection on observables (or strong ignorability) assumption, we show that the effective estimand is a weighted average of covariate-specific rank-ATE, in a similar spirit to \cite{angrist1995estimating} and \cite{borusyak2024negative} for linear regressions. 

Moving beyond the binary treatment, we consider general multi-valued or continuous treatments and allow for transformation of the treatment variable in the regression. This includes the rank-rank regression \citep[e.g.][]{dahl2008association, chetty2014land, chetverikov2023inference} and special cases of binscatter regression \citep{cattaneo2024binscatter} as special cases. When the treatment is randomized, we prove that the effective OLS estimand is a weighted average of rank-ATEs that compare the potential outcomes at each level of treatment. Since the weights depend on the scale of the transformed treatment, the estimates for different treatment variables are generally incomparable. We propose a data-driven normalization procedure that restrict the estimand in the same range regardless of the choice of variables. 

For other econometric methods, the effective estimands and causal interpretations are more delicate. To avoid mathematical complications, we focus on the case where the outcomes are transformed into ranks and the treatment variable is binary.  We briefly summarize our findings below.
\begin{itemize}
\item (2SLS) We introduce the rank-local average treatment effect (rank-LATE) as the rank-ATE for compliers, which can be viewed as a natural extension of LATE \citep{imbens1994identification}. When the instrument variable (IV) is binary and randomized, under the standard exclusion restriction and monotonicity assumptions, we show that the 2SLS does not recover the rank-LATE and its estimand can have an opposite sign. Further, we show that the 2SLS does not enjoy reference robustness and ranking based on neither the treated nor control group recovers the rank-LATE. To address this issue, we propose a new ranking procedure based on an estimated potential outcome distribution for compliers and prove that it identifies the rank-2SLS without additional assumptions.
\item (DiD) For simplicity, we focus on the case with one pre-treatment period and one post-treatment period. We start by introducing a natural extension of the standard parallel trend assumption. The rank version is non-nested with the standard version or the counterpart for distributional DiD \citep{havnes2015universal, roth2023parallel}, but implied by the identifying assumption for change-in-changes (CiC, \cite{athey2006identification}) or the copula stability assumption \citep{ghanem2023evaluating}. Nevertheless, we prove that the DiD estimator on outcome ranks does not identify the rank analogue of the average treatment effect on the treated (rank-ATT); instead, the effective estimand is hard to interpret and not necessarily same-signed as the rank-ATT. We then propose a counterfactual ranking method and show the resulting DiD estimator recovers the rank-ATT under the assumptions for CiC \citep{athey2006identification} or distributional DiD \citep{roth2023parallel}. 
\item (RDD) For level outcomes, the estimand is the ATE at the cutoff for sharp RDD \citep{imbens2008regression, lee2010regression}. We prove that, when the outcomes are ranked based on the entire sample, the RDD estimator fails to recover the rank-ATE at the cutoff due to contamination of points further away from the cutoff. By strengthening the mean continuity into a distributional continuity assumption  \citep{imbens2008regression}, we propose a local ranking method that identifies the rank-ATE at the cutoff for sharp RDD without additional assumptions.
\end{itemize}

We conclude the paper by discussing a limitation of the rank-ATE: unlike the standard ATE, it cannot be expressed as an expected individual-level treatment effect. \cite{chen2023logs} argued that the latter is a desirable property for the purpose of causal interpretations. Motivated by this concern, we propose an alternative partially identified estimand. The sharp identified set can be estimated by the standard plug-in methods \citep{fan2010sharp} or the robust dual methods \citep{ji2023model}. 

\subsection{Notation and preliminaries}
Suppose we have $n$ units and, for unit $i\in \{1, \ldots, n\}$, $Y_i$ denotes the outcome of interest, $W_i$ denotes the treatment variable which can be binary or continuous, and $X_i\in \R^{d}$ denotes the vector of covariates. Let $\RY_i$ denote the rank of $Y_i$ among $\{Y_1, \ldots, Y_n\}$, i.e.,
\begin{equation}\label{eq:rank}
  \RY_i = \sum_{j=1}^{n}I(Y_j \le Y_i).
\end{equation}
Let $\hat{F}_{Y,n}$ be the empirical cumulative distribution function (CDF) of $\{Y_1, \ldots, Y_n\}$, i.e.,
\[\hat{F}_{Y,n}(y) = \frac{1}{n}\sum_{j=1}^{n}I(Y_j \le y).\]
Then we can rewrite $\RY_i$ as
\begin{equation}\label{eq:rank_ecdf}
\RY_i = n\hat{F}_{Y,n}(Y_i).
\end{equation}
Throughout the note we assume that $(Y_i, W_i, X_i)$ are i.i.d. where $F_{Y}$ denotes the marginal distribution of $Y_i$, $F_{Y\mid w}$ the conditional distribution of $Y_i$ given $W_i = w$, and $F_{Y\mid w, x}$ the conditional distribution of $Y_i$ given $W_i = w$ and $X_i = x$. When no confusion can arise, we will use $(Y, W, X)$ to denote a generic independent draw from the same distribution as $(Y_i, W_i, X_i)$. Throughout we adopt the standard probability notation $O(\cdot), o(\cdot), O_{\P}(\cdot),$ $o_{\P}(\cdot)$. 

An immediate yet crucial implication of the celebrated Dvoretzky–Kiefer–Wolfowitz inequality \citep{dvoretzky1956asymptotic, massart1990tight} is that
\begin{equation}\label{eq:rank_approx}
  \max_{i}\left|\frac{\RY_i}{n} - F_Y(Y_i)\right|\le \sup_{y\in \R}|\hat{F}_{Y,n}(y) - F_{Y}(y)| = O_\P\lb \frac{1}{\sqrt{n}}\rb.
  \end{equation}
This suggests, and we will make it rigorous, that $\RY_i$ can be replaced by $F_Y(Y_i)$ when the goal is to prove consistency of a parameter estimate. For inferential purposes like hypothesis testing or confidence intervals, more sophiscated techniques would be involved \citep{chetverikov2023inference}.

\section{Rank-ATE: a primitive causal estimand}
Let $F_1$ and $F_0$ be the distributions of a scalar outcome under two treatment conditions. Most point identified treatment effects are functionals of $(F_1, F_0)$ that measure the distributional discrepancy through certain summaries.\footnote{Causal effects as functionals of the joint distribution of potential outcomes are only partially identified without further assumptions since the potential outcomes are never observed simultaneously \citep[e.g.][]{fan2010sharp,ji2023model}.} For example, the standard ATE is the mean difference and the logarithm of the ratio ATE (or the Poisson regression estimand \citep{silva2006log,wooldridge2010econometric}) is the difference of log-means. Here, we define the rank-ATE $\rate(F_1, F_0)$ as follows: 
\begin{equation}
  \label{eq:rank-ATE}
  \rate(F_1, F_0) = \P(Z_1 \ge Z_0) - \frac{1}{2}, \quad \text{where }Z_1 \sim F_1, Z_0 \sim F_0, Z_1\indep Z_0.
\end{equation}
The rank-ATE always takes values in $[-1/2, 1/2]$. Clearly, $\rate(F_1, F_0)$ measures the probability that a randomly drawn sample from $F_1$ is larger than or equal to a randomly drawn sample from $F_0$. Unlike the ATE that requires the existence of finite first moments, the rank-ATE is always well-defined. An equivalent expression for $\rate(F_1, F_0)$ is given by 
  \begin{equation}
    \label{eq:rate_alternative}
    \rate(F_1, F_0) = \E_{Z_1\sim F_1}[F_0(Z_1)] - \frac{1}{2}.
  \end{equation}
 We illustrate this estimand through the following examples.
  \begin{example}\label{ex:symmetric_binary}
    If $Z_1 = \mu_1 + \nu_1$ and $Z_0 = \mu_0 + \nu_0$ for some symmetrically distributed $\nu_1$ and $\nu_0$ that are not necessarily identically distributed,
    \begin{equation}\label{eq:AUC_example1_prob}
      \P(Z_1\ge Z_0) = \P(\nu_1 - \nu_0\ge \mu_0 - \mu_1).
      \end{equation}
      Since $\nu_1$ and $\nu_0$ are independent and symmetrically distributed, so is the difference $\nu_1 - \nu_0$. As a result,
      \begin{equation}\label{eq:AUC_example1_sign}
        \rate(F_1, F_0) > 0 \Longleftrightarrow \mu_1 > \mu_0,
      \end{equation}
      Thus the rank-ATE has the same sign as the standard ATE, if exists. 
  \end{example}
  \begin{example}\label{ex:id_binary}
    If $Z_1 = \mu_1 + \nu_1$ and $Z_0 = \mu_0 + \nu_0$ for some identically distributed $\nu_1$ and $\nu_0$ with arbitrary distributions. Then $\nu_1 - \nu_0$ is symmetrically distributed. By \eqref{eq:AUC_example1_prob}, the rank-ATE and ATE, if exists, have the same signs. 
  \end{example}
  \begin{example}\label{ex:stochastic_monotonicity_binary}
    If $Z_1$ (first-order) stochastically dominates $Z_0$, then there exists a random variable $Z'_0$ that is equally distributed as $Z_0$ and $Z'_0\le Z_1$ almost surely. Thus,
    \[\rate(F_1, F_0) \ge \P(Z'_0 \ge Z_0) - \frac{1}{2} = 0.\]
    Thus, the rank-ATE also has the same sign as the ATE, if exists. 
  \end{example}

  The rank-ATE has the following nice properties.
  \begin{itemize}
  \item (Linearity) for any distributions $F_1, F_1', F_0, F_0'$    and $\alpha_1, \alpha_0\in (0, 1)$,
    \[\rate(\alpha_1 F_1 + (1 - \alpha_1)F_1', F_0) = \alpha_1\rate(F_1, F_0) + (1 - \alpha_1)\rate(F_1', F_0),\]
    and
    \[\rate(F_1, \alpha_0 F_0 + (1 - \alpha_0)F_0') = \alpha_0\rate(F_1, F_0) + (1 - \alpha_0)\rate(F_1, F_0').\]
  \item (Anti-symmetry) when $F_1$ and $F_0$ are both absolutely continuous,
    \[\rate(F_0, F_1) = -\rate(F_1, F_0).\]
    In particular, $\rate(F, F) = 0$ for any continuous distribution $F$.
  \item (Invariance) for any strictly increasing function $f$,
    \[\rate(F_0, F_1) = \rate(F_0\circ f^{-1}, F_1\circ f^{-1}).\]
  \item (Partial additivity) for any $\zeta\in [0, 1]$,
    \[\rate(F_1, \zeta F_1 + (1 - \zeta)F_0) - \rate(F_0, \zeta F_1 + (1 - \zeta)F_0) = \rate(F_1, F_0).\]
  \end{itemize}
  We call the last property partial additivity because it only holds when the second argument of the first two terms is in the class $\{\zeta F_1 + (1 - \zeta)F_0: \zeta \in [0, 1]\}$. Note that the ATE satisfies all but invariance, which is a desirable property especially when the variables do not have a natural unit or it is unclear what transformation is most scientifically meaningful \citep[e.g.][]{roth2023parallel}. The first three properties can be proved directly by definition and the partial additivity is proved in Appendix \ref{app:proofs}. 

 The rank-ATE has been studied in different contexts. In particular, it is the effective estimand of the Mann-Whitney test, also known as the Wilcoxon rank-sum test \citep[e.g.][]{vermeulen2015increasing}. It was also shown to be equivalent to the Area Under the Receiver Operating Characteristic curve (AUROC, also known as AUC), one of the most important measures in evaluating classification algorithms \citep{hanley1982meaning}. 

\section{Rank-OLS with a binary treatment}\label{sec:rank_OLS_binary}
In this section, we study the case where $W_i$ is binary. Let $(n_1, n_0)$ denote the size of the treated and control groups, respectively, $\pi = \P(W_i = 1)$ the marginal treatment intensity, and $\pi(z) = \P(W_i = 1\mid X_i = z)$ the propensity score. Furthermore, we let $(Y_i(1), Y_i(0))$ denote the pair of potential outcomes for unit $i$. Throughout the section we assume the stable unit treatment value assumption (SUTVA, \cite{rubin1974estimating}):
\begin{assumption}\label{as:SUTVA}
  For all $i = 1,\ldots, n$,
  \[Y_i = Y_i(1)W_i + Y_i(0)(1 - W_i).\]
\end{assumption}
Let $F_{Y(1)}$ and $F_{Y(0)}$ denote the marginal distributions of $Y(1)$ and $Y(0)$, respectively. Clearly,
\begin{equation}
  \label{eq:mixture}
  F_Y = \pi F_{Y(1)} + (1 - \pi)F_{Y(0)}.
\end{equation}
We further assume that $Y(1)$ and $Y(0)$ are continuous variables:
\begin{assumption}\label{as:continuous_Y}
  $F_{Y(1)}$ and $F_{Y(0)}$ are both absolutely continuous with respect to the Lebesgue measure on $\R$. 
\end{assumption}
Even if $Y(1)$ and $Y(0)$ are discrete, one can perturb $Y(1)$ and $Y(0)$ by a random noise uniformly distributed  on $[-\eps, \eps]$ for some sufficiently small $\eps$. This is equivalent to random tie-breaking in practice considering that all variables are discretized by rounding. 

\subsection{Rank-OLS without covariates under random assignments}\label{eq:rank_OLS}
As a warm-up, we consider the rank-OLS estimator without covariates, defined as follows:
\begin{equation}\label{eq:OLS_nocov}
    \hat{\beta}_{\nox} = \argmin \frac{1}{n}\sum_{i=1}^{n}\lb\frac{\RY_i}{n} - \beta_0 - \beta W_i\rb^2.
  \end{equation}
Above, the ranks are normalized by $n$ for ease of interpretation, as evidenced later, and the subscript indicates that no covariate is used. We study the property of $\hat{\beta}_{\nox}$ when assignments are completely random as in randomized experiments or quasi-experiments:
  \begin{assumption}\label{as:randomized}
    For all $i = 1,\ldots, n$,
    \[\P(W_i = 1\mid Y_i(1), Y_i(0), X_i) = \pi\in (0, 1).\]
  \end{assumption}
  Then we can prove the following result that the rank-OLS identifies the rank-ATE. 
  \begin{theorem}\label{thm:binary_noz_randomized}
    Under Assumptions \ref{as:SUTVA}, \ref{as:continuous_Y}, and \ref{as:randomized}, as $n\rightarrow\infty$,
    \begin{equation}\label{eq:beta_star}
      \hat{\beta}_{\nox} \stackrel{p}{\rightarrow}\rate(F_{Y(1)}, F_{Y(0)}).
      \end{equation}
    \end{theorem}

  An intriguing question is how the estimand changes with the reference population for ranking. In pricinple, one could rank each outcome among all units or among treated or control units. Specifically, define two other ranks as follows:
  \[\RY_{i,w} = \sum_{j:W_j=w}I(Y_j\le Y_i), \quad w = 0, 1,\]
  and consider the corresponding rank-OLS:
  \[\hat{\beta}_{\nox, w} = \argmin \frac{1}{n}\sum_{i=1}^{n}\lb\frac{\RY_{i,w}}{n_w} - \beta_0 - \beta W_i\rb^2, \quad w = 0, 1.\]
  Note that we use different normalizations because the range of $\RY_{i,1}$ and $\RY_{i,0}$ are different. Surprisingly, both $\hat{\beta}_{\nox, 1}$ and $\hat{\beta}_{\nox, 0}$ are equal to $\hat{\beta}_{\nox}$ up to a small factor.
  \begin{theorem}\label{thm:reference_robustness}
    Assuming no ties among $Y_i$s, 
    \[\hat{\beta}_{\nox, 1} = \hat{\beta}_{\nox} + \frac{1}{2n_1}, \quad \hat{\beta}_{\nox, 0} = \hat{\beta}_{\nox} - \frac{1}{2n_0}.\]
    As a result, if $n_1, n_0 \rightarrow \infty$ as $n\rightarrow \infty$,
    \[\hat{\beta}_{\nox, 1}\stackrel{p}{\rightarrow} \rate(F_{Y(1)}, F_{Y(0)}), \quad \hat{\beta}_{\nox, 0}\stackrel{p}{\rightarrow} \rate(F_{Y(1)}, F_{Y(0)}).\]
  \end{theorem}
  
  \subsection{Rank-OLS with covariates under random assignments}\label{subsec:OLS_binary_cov}
  In practice, covariates are often added into the regression for the purpose of adjusting for confounders or improving statistical efficiency. We consider two versions of rank-OLS, one without interactions and one with interactions between the treatment and covariates:
  \begin{equation}\label{eq:OLS}
    \hat{\beta} = \argmin \frac{1}{n}\sum_{i=1}^{n}\lb\frac{\RY_i}{n} - \beta_0 - \beta W_i - X_i^\top\eta\rb^2,
  \end{equation}
  and 
  \begin{equation}\label{eq:OLS_inter}
    \hat{\beta}_{\inter} = \argmin \frac{1}{n}\sum_{i=1}^{n}\lb\frac{\RY_i}{n} - \beta_0 - \beta W_i - X_i^\top\eta - W_i (X_i - \bar{X})^\top \td{\eta}\rb^2,
  \end{equation}
  where $\bar{X} = (1/n)\sum_{i=1}^{n}X_i$. The former estimator is widely used in applied work and the latter estimator is much less so. The latter estimator is equivalent to running separate regressions on the treated and control units with centered covariates $X_i - \bar{X}$ and contrasting the intercepts. Our motivation to consider $\hat{\beta}_{\inter}$ is from the regression adjustment literature \citep[e.g.][]{freedman2008regression, imbens2009recent,
  lin2013agnostic, li2020rerandomization, lei2021regression, negi2021revisiting}.

  The following result shows that the OLS estimand remains the same under the following mild regularity condition. 
  \begin{assumption}\label{as:second_moment}
    $0 < \lambda_{\min}(\E[XX']) \le \lambda_{\max}(\E[XX']) < \infty$ where $\lambda_{\min}$ and $\lambda_{\max}$ denote the minimal and maximal eigenvalues. 
  \end{assumption}
  
  \begin{theorem}\label{thm:binary_z_randomized}
    In the setting of Theorem \ref{thm:binary_noz_randomized} and assuming Assumption \ref{as:second_moment}, as $n\rightarrow\infty$, 
    \begin{equation*}
      \hat{\beta} \stackrel{p}{\rightarrow}\rate(F_{Y(1)}, F_{Y(0)}),\quad \text{and }\quad \hat{\beta}_{\inter} \stackrel{p}{\rightarrow}\rate(F_{Y(1)}, F_{Y(0)}).
    \end{equation*}
  \end{theorem}

\subsection{Rank-OLS with covariates under confounded assignments}\label{subsec:OLS_binary_confounded}

In this subsection, we relax the assumption of random assignments and assume the propensity score $\pi(z)$ varies with $z$. To make progress, we assume strong ignorability or selection on observables:
\begin{assumption}\label{as:strong_ignorability}
  $Y_i(1), Y_i(0)\indep W_i\mid X_i$.
\end{assumption}
In addition, we assume strict overlap or positivity \citep[e.g.][]{crump2009dealing,
d2021overlap, lei2021distribution}:
\begin{assumption}\label{as:overlap}
  There exist constants $c \in (0, 1)$ such that $\pi(z) \in [c, 1-c]$ for all $z$. 
\end{assumption}
We state a generic result that shows $\hat{\beta}$ converges to a weighted average of conditional rank-ATEs. The proof is presented in Appendix \ref{app:proofs}.
\begin{theorem}\label{thm:weighted_average}
  Let $\td{\pi}(X)$ be the population projection of $\pi(X)$ onto the linear space of $X$, i.e.,
  \[\td{\pi}(X) = X^\top\omega^{*}, \quad \text{where }\omega^{*} = \argmin_{\omega}\E[(\pi(X) - X^\top \omega)^2] = (\E[XX'])^{-1}\E[X\pi(X)].\]
  Assume one of the following assumptions hold:
  \begin{enumerate}[(a)]
  \item $\pi(X) = X^\top \omega^{*}$;
  \item $\E[F_Y(Y(0))\mid X] = X^\top \eta^{*}$ for some $\eta^{*}\in \R^{d}$.
  \end{enumerate}
  Then
  \[\hat{\beta} \stackrel{p}{\rightarrow}\E\left[\frac{w(X)}{\E[w(X)]}\rate(F_{Y(1)\mid X}, F_{Y(0)\mid X})\right], \quad \text{where }w(X) = \pi(X)(1 - \td{\pi}(X))\]
\end{theorem}
The condition (a) and (b) are similar to the ones in \cite{borusyak2024negative}. In particular, the condition (b) is the analogue of the classical assumption of saturated specification \citep{angrist1995estimating}. When $\td{\pi}(X)\in [0, 1]$, the effective estimand of $\hat{\beta}$ is a convex average of conditional rank-ATEs.

\section{Rank-OLS with a general treatment}
In this section, we consider a general treatment $W_i \in \W \subset\R$, where $\W$ can be a finite or infinite set. Define the rank of $W_i$ as 
\[\RW_i = \frac{1}{n}\sum_{j=1}^{n}I(W_j\le W_i).\]
Let $F_W$ denote the marginal distribution of $W$, $Y(w)$ the potential outcome for dose $w$, and $F_{Y(w)}$ the marginal distribution of $Y(w)$. We make the following assumptions as analogues of Assumptions \ref{as:SUTVA} and \ref{as:continuous_Y}.
\begin{assumption}\label{as:SUTVA_cont}
  For all $i = 1,\ldots, n$, $Y_i = Y_i(W_i)$. 
\end{assumption}
\begin{assumption}\label{as:continuous_Y_cont}
  $F_{Y(w)}$ is absolutely continuous for any $w\in \R$. 
\end{assumption}

\subsection{Rank-OLS under random assignments}
We consider the rank-OLS estimator with a transformed treatment:
\begin{equation}\label{eq:OLS_general}
  \hat{\beta}_{h_n} = \argmin \frac{1}{n}\sum_{i=1}^{n}\lb\frac{\RY_i}{n} - \gamma_0 - \gamma h_n\lb W_i\rb - X_i^\top\eta\rb^2,
\end{equation}
where $h_n$ is any function that potentially depends on data (and hence is potentially random). This nests a broad class of regressions in applied work. For example, when $h_n$ is the identity mapping, $\hat{\beta}_{h_n}$ recovers $\hat{\beta}$ studied in Section \ref{subsec:OLS_binary_cov}. Another widely-used class of regressions that is nested in \eqref{eq:OLS_general} is the rank-rank regressions:
\begin{equation}\label{eq:rank-rank}
  \hat{\gamma}_{g_n} = \argmin \frac{1}{n}\sum_{i=1}^{n}\lb\frac{\RY_i}{n} - \gamma_0 - \gamma g_n\lb \frac{\RW_i}{n}\rb - X_i^\top\eta\rb^2.
\end{equation}
In particular, $\hat{\gamma}_{g_n}$ is identical to $\hat{\beta}_{h_n}$ with $h_n = g_n\circ \hat{F}_{W}$ where $\hat{F}_W$ is the empirical CDF of $W_i$s. When $g_n$ is the identity mapping, \eqref{eq:rank-rank} is the standard rank-rank regression \citep[e.g.][]{chetverikov2023inference}. It also nests regressions that dichotomize $W_i$ by choosing $g_n(r) = I(r > 1/2)$ or those that coarsen $W_i$ into quartiles/quintiles/deciles by choosing a stepwise function.

It is known that, when $Y_i$ is continuous, $g_n(r) = r$, and no covariate is included, $\hat{\gamma}_{\nox}$ is equivalent to the Spearman's rank correlation \citep[e.g.][]{chetverikov2023inference}. The estimand can be equivalently expressed as $\Cor(F_W(W), F_Y(Y))$ where $\Cor$ denotes the Pearson correlation coefficient. However, the causal interpretation of this estimand is underexplored. We first derive an expression of the estimand in terms of potential outcomes when $W_i$ is randomly assigned:
\begin{assumption}\label{as:randomized_cont}
  For all $i = 1,\ldots, n$,
  \[W_i \mid (\{Y_i(w): w\in \W\}, X_i) \sim F_W.\]
\end{assumption}

To study the limit of $\hat{\beta}_{h_n}$, we impose the following assumption on limited dependence of $h_n$ on data. It is easy to check that the aforementioned examples all satisfy the assumption.
  \begin{assumption}\label{as:limit_h}
    There exists a non-stochastic function $h: [0, 1]\mapsto \R$ such that, as $n\rightarrow \infty$,
    \[\sum_{i=1}^{n}\lb h_n(W_i) - h(W_i)\rb^2  = O_\P(1).\]
  \end{assumption}
  For rank-rank OLS estimators, by \eqref{eq:rank_approx}, Assumption \ref{as:limit_h} so long as
  \[\sum_{i=1}^{n}\lb g_n(W_i) - g(W_i)\rb^2  = O_\P(1),\]
  for some non-stochastic function $g$.
  
  \begin{theorem}\label{thm:general_noz_randomized}
    Under Assumptions \ref{as:SUTVA_cont}, \ref{as:continuous_Y_cont}, \ref{as:randomized_cont}, and \ref{as:limit_h}, the rank-OLS estimator defined in \eqref{eq:OLS_general} has the following limit as $n\rightarrow\infty$:
       \begin{equation}\label{eq:gamma_star_2}
      \hat{\beta}_{h_n} \stackrel{p}{\rightarrow} \beta_h^{*}\triangleq \E\left[\frac{b(W, \td{W})}{\E[b^2(W, \td{W})]}\rate\lb F_{Y(W)}, F_{Y(\td{W})}\rb\right],
    \end{equation}
    where $(W, \td{W})$ are two independent draws from $F_W$ and
    \begin{equation}
      \label{eq:b_weight}
      b(W, \td{W}) = (h(W) - h(\td{W})) I(W > \td{W}).
    \end{equation}
  \end{theorem}
  The proof is presented in Appendix \ref{app:proofs}. Clearly, when $h$ is monotonely increasing, $b(W, \td{W})\ge 0$. In this case, Theorem \ref{thm:general_noz_randomized} implies that the effective estimand is a weighted sum of $\rate\lb F_{Y(w)}, F_{Y(\td{w})}\rb$ for all dose pairs $(w, \td{w})$ with $w\ge \td{w}$ with non-negative weights. 
  
  To further illustrate the estimand, we consider the following examples as extensions of Example \ref{ex:symmetric_binary} to \ref{ex:stochastic_monotonicity_binary}.
    \begin{example}\label{ex:symmetric_general}
    If $Y(w) = \mu(w) + \nu(w)$ where $\nu(w)$ is symmetrically distributed for each $w$, then $\beta_h^{*} \ge 0$ if $\mu(w)$ is weakly increasing in $w$. For example, the conditions are satisfied when $Y(w)$ is a Gaussian process with weakly increasing mean function. 
  \end{example}
  \begin{example}\label{ex:id_general}
    If $Y(w) = \mu(w) + \nu(w)$ where $\nu(w)$ and $\nu(\td{w})$ are identically distributed for each pair $(w, \td{w})$. Then $\beta_h^{*} \ge 0$ if $\mu(w)$ is weakly increasing in $w$ whenever $\mu(w)$ is weakly increasing.
  \end{example}
  \begin{example}\label{ex:stochastic_monotonicity_general}
    If $Y(w)$ is weakly stochastically increasing in $w$, then $\beta_h^{*}\ge 0$.
  \end{example}
  We further provide examples with different $h_n$. 
  \begin{example}\label{ex:estimand_binary_w}
    If $W_i$ is binary and $h_n(w) = w$, $\hat{\beta}_{h_n} = \hat{\beta}_{\nox}$. In this case, $F_W(W) = W + (1 - \pi)(1 - W)$ and thus,
    \[b(W, \td{W}) = I(W = 1, \td{W} = 0) = W(1 - \td{W}).\]
    Then
    \[\beta_{h}^{*} = \frac{\E[W(1 - \td{W})\rate(F_{Y(1)}, F_{Y(0)})]}{\E[W(1 - \td{W})]} = \rate(F_{Y(1)}, F_{Y(0)}).\]
  \end{example}
  \begin{example}\label{ex:estimand_dichotomization}
    Suppose $W_i$ is continuous and $g_n(r) = I(r > 1/2)$ in the rank-rank OLS. Let $m_{W}$ be the median of $W_i$. Then
    \[b(W, \td{W}) = I(W > m_{W} > \td{W}).\]
    Under Assumption \ref{as:randomized_cont}, $F_{Y\mid W} = F_{Y(W)\mid W}$. Thus,
    \begin{align*}
      &\beta_{h}^{*} = 4\E[\rate(F_{Y(W)}, F_{Y(\td{W})})I(W > m_{W})I(\td{W} < m_{W})]\\
      & = 4\E[\rate(F_{Y\mid W}, F_{Y\mid \td{W}})I(W > m_{W})I(\td{W} < m_{W})]\\
      & = \rate(F_{Y\mid W > m_{W}}, F_{Y\mid W < m_{W}}),
    \end{align*}
    where the last line invokes linearity of $\rate(\cdot, \cdot)$.
  \end{example}

  \subsubsection{Normalizing $h_n$ for comparability across different treatment variables}
  To enable comparison of the OLS estimates across different treatment variables, we need to enforce the estimands to have the same scale. It suffices to ensure that $\beta^{*}_{h}$ is a convex average of pairwise effects $\{\rate\lb F_{Y(w)}, F_{Y(\td{w})}\rb: w > \td{w}\}$, in which case $\beta^{*}_{h}\in [-1/2, 1/2]$. To achieve this, we can rescale $h_n$ such that
  \begin{equation}\label{eq:convex_average}
    \E[b(W, \td{W})] = \E[b^2(W, \td{W})].
  \end{equation}
  For any given $h_n$, we can replace it by
  \[\frac{\E[(h_n(W) - h_n(\td{W}))I(W > \td{W})]}{\E[(h_n(W) - h_n(\td{W}))^2I(W > \td{W})]}\cdot h_n.\]
  Both the numerator and denominator can be estimated by U-statistics:
  \begin{equation}\label{eq:transform_hn}
    \td{h}_{n}(w) = \hat{\kappa} h_n(w), \quad \text{where }\hat{\kappa} = \frac{\sum_{i\neq j}(h_n(W_i) - h_n(W_j))I(W_i > W_j)}{\sum_{i\neq j}(h_n(W_i) - h_n(W_j))^2I(W_i > W_j)}.
  \end{equation}
  For example, when $h_n(w) = w$ as in the case of binary treatment, $\hat{\kappa} = 1$. When $h_n(w) = \hat{F}_W(w)$,
  \begin{align*}
    \hat{\kappa} &= \frac{(1/n)\sum_{i>j}(i - j)}{(1/n^2)\sum_{i>j}(i - j)^2} = \frac{(1/n)\sum_{i} i(i-1)/2}{(1/n^2)\sum_{i}(i-1)i(2i-1)/6} = \frac{(n+1)(n-1)/6}{(n+1)(n-1)/12} = 2.
  \end{align*}
  Thus, we need to replace $\RW_i / n$ by $2\RW_i / n$.
  
  \section{Other econometric methods}\label{sec:other_reg}
  To ease exposition and highlight the main point, we limit our discussion to binary treatments and exclude the consideration of covariates. All technical proofs are presented in Appendix \ref{app:proofs}.
  
  \subsection{Rank-2SLS with a binary IV}
  Let $Z_i$ be a valid binary instrumental variable, $(W_i(1), W_i(0))$ be the potential treatment assignments, and $\{Y_i(w, z): w, z\in \{0, 1\}\}$ be the potential outcomes. We make the standard IV assumptions \citep{angrist1996identification}.
  \begin{assumption}\label{as:IV}
    For all $i = 1,\ldots,n$,
    \begin{enumerate}[(a)]
    \item (exclusion restriction) $Y_i(w, 1) = Y_i(w, 0) \triangleq Y_i(w)$ for $w = 0,1$;
    \item (ignorability) $Z_i\indep (W_i(1), W_i(0), Y_i(1), Y_i(0))$;
    \item (monotonicity) $W_i(1) \ge W_i(0)$ almost surely;
    \item (relevance) $\P(W=1\mid Z=1) > \P(W=1\mid Z=0)$.
    \end{enumerate}
  \end{assumption}

 Further, we let $G_i$ denote the type of the unit $i$ (i.e., always-taker, never-taker, and complier) with
  \[G_i = \left\{
      \begin{array}{cc}
        \ta & W_i(1) = W_i(0) = 1\\
        \tn & W_i(1) = W_i(0) = 0\\
        \tc & W_i(1) = 1, W_i(0) = 0.
      \end{array}
    \right..\]
  Throughout this secvtion we assume that $(Z_i, W_i(1), W_i(0), Y_i(1), Y_i(0), G_i)$ are i.i.d. and denote by $(Z, W(1), W(0), Y(1), Y(0), G)$ a generic draw. For each $g\in \{\ta, \tn, \tc\}$ and $w\in\{0, 1\}$, let $F_{Y(w)\mid g}$ denote the conditional distribution of $Y(w)$ given $G = g$.

  Let $\hat{\beta}_{\sls}$, $\hat{\beta}_{\sls, 1}$, $\hat{\beta}_{\sls, 0}$ denote the 2SLS estimators applied to $\RY_i / n$, $\RY_{i,1}/n_1$, and $\RY_{i,0}/n_0$, respectively. First, we derive the effective estimand of 2SLS applied to different types of ranks.
  \begin{theorem}\label{thm:2sls}
    Under Assumptions \ref{as:SUTVA}, \ref{as:continuous_Y} and \ref{as:IV}, as $n\rightarrow \infty$,
    \[\hat{\beta}_{\sls}\stackrel{p}{\rightarrow} \rate(F_{Y(1)\mid \tc}, F_{Y}) - \rate(F_{Y(0)\mid \tc}, F_{Y}),\]
    \[\hat{\beta}_{\sls,1}\stackrel{p}{\rightarrow} \rate(F_{Y(1)\mid \tc}, F_{Y\mid W=1}) - \rate(F_{Y(0)\mid \tc}, F_{Y\mid W=1}),\]
    \[\hat{\beta}_{\sls,0}\stackrel{p}{\rightarrow} \rate(F_{Y(1)\mid \tc}, F_{Y\mid W=0}) - \rate(F_{Y(0)\mid \tc}, F_{Y\mid W=0}),\]
Without further assumptions, the above estimands are mutually different.
  \end{theorem}
  Theorem \ref{thm:2sls} suggests that, unlike the rank-OLS estimator, the rank-2SLS estimator is no longer reference-robust. In addition, none of these estimands are as easily interpretable as the rank-ATE.

  Since LATE is the mean difference of potential outcomes for compliers, it is natural to consider the rank-LATE, defined as the rank-ATE for compliers: 
  \begin{equation}
    \label{eq:rate_IV}
    \text{rank-LATE} \triangleq \rate(F_{Y(1)\mid \tc}, F_{Y(0)\mid \tc}).
  \end{equation}
The following result shows that rank-LATE can be obtained by ranking the outcomes among the compliers.
  \begin{theorem}\label{thm:complier_ranking}
    Fix any $\zeta \in [0, 1]$. Let $\hat{F}_{Y(w)\mid \tc}$ be consistent estimates of $F_{Y(w)\mid \tc}$ in the sense of \eqref{eq:rank_approx}. Further let $\hat{\beta}_{\sls, \zeta, \tc}$ denote the 2SLS estimator applied to $\zeta \hat{F}_{Y(1)\mid \tc} + (1 - \zeta)\hat{F}_{Y(0)\mid \tc}$. Then, as $n\rightarrow \infty$, 
    \[\hat{\beta}_{\sls, \zeta, \tc}\stackrel{p}{\rightarrow} \rate(F_{Y(1)\mid \tc}, F_{Y(0)\mid \tc}).\]
  \end{theorem}
  Theorem \ref{thm:complier_ranking} implies that this new 2SLS estimator identifies the rank-LATE and enjoys reference robustness as the rank-OLS estimator under randomized assignments.

  The next result shows that $F_{Y(1)\mid \tc}$ and $F_{Y(0)\mid \tc}$ can be identified under no additional assumptions. 

  \begin{theorem}\label{thm:complier_ranking_estimator}
    For any $y\in \R, w, z\in \{0, 1\}$, let
  \[F_{wz}(y) = \P(Y\le y\mid W = w, Z = z).\]
  Then, in the setting of Theorem \ref{thm:2sls},
    \begin{equation*}
    F_{Y(1)\mid \tc} = \frac{\pi_{\ta} + \pi_{\tc}}{\pi_{\tc}}\lb F_{11} - \frac{\pi_{\ta}}{\pi_{\ta} + \pi_{\tc}}F_{10}\rb,
  \end{equation*}
    and
  \begin{equation*}
    F_{Y(0)\mid \tc} = \frac{\pi_{\tn} + \pi_{\tc}}{\pi_{\tc}}\lb F_{00} - \frac{\pi_{\tn}}{\pi_{\tn} + \pi_{\tc}}F_{01}\rb.
  \end{equation*}
  \end{theorem}
 
  For each $w,z\in \{0,1\}$, $F_{wz}$ can be consistently estimated by the empirical CDF of $Y_i$s among the units with $W_i = w, Z_i = z$. Classical IV theory \citep{imbens1994identification} implies that
  \[\frac{n_{10}}{n_0}\stackrel{p}{\rightarrow} \pi_{\ta}, \quad \frac{n_{01}}{n_1}\stackrel{p}{\rightarrow} \pi_{\tn}, \quad \frac{n_{11}}{n_1} - \frac{n_{10}}{n_0}\stackrel{p}{\rightarrow} \pi_{\tc}\]
  where $n_{wz}$ is the number of units with $W_i = w, Z_i = z$ and $n_{z}$ is the number of units with $Z_i = z$.
  
\subsection{Rank-DiD}
We focus on the standard DiD setting with one pre-treatment period, indexed by $0$, and one post-treatment period, indexed by $1$. For each $t=0, 1$, let $(Y_{it}(1), Y_{it}(0))$ denote the potential outcomes of unit $i$ at time $t$ and $W_i$ denote the treatment status at time $1$. Further let $\RY_{it}$ denote the rank of $Y_{it}$ among $\{Y_{jt}: j=1,\ldots,n\}$. Clearly,
\[\RY_{it} = n\hat{F}_{Y_t,n}(Y_{it}),\]
where $\hat{F}_{Y_t,n}$ is the empirical CDF of $\{Y_{jt}: j=1,\ldots,n\}$. 

Throughout the section we assume $(Y_{i0}(1), Y_{i0}(0), Y_{i1}(1), Y_{i1}(0), W_i)$ are i.i.d. and denote by $(Y_0(1), Y_0(0), Y_1(1), Y_1(0), W)$ a generic draw. Similar to the cross-sectional case, we make the following assumptions.
\begin{assumption}\label{as:SUTVA_did}
  For all $i = 1,\ldots, n$,
  \[Y_{i1} = Y_{i1}(1)W_i + Y_{i1}(0)(1 - W_i), \quad Y_{i0} = Y_{i0}(0).\]
\end{assumption}
\begin{assumption}\label{as:continuous_Y_did}
  $F_{Y_t(w)}$ is absolutely continuous with respect to the Lebesgue measure on $\R$ for any $w,t\in \{0, 1\}$.
\end{assumption}

The standard rank-DiD estimator is defined as
\begin{equation}
  \label{eq:did}
  \hat{\beta}_{\did} = \frac{1}{n_1}\sum_{i: W_i=1}\lb\frac{\RY_{i1}}{n} - \frac{\RY_{i0}}{n}\rb - \frac{1}{n_0}\sum_{i: W_i=0}\lb\frac{\RY_{i1}}{n} - \frac{\RY_{i0}}{n}\rb.
\end{equation}
Without rank transformation, the usual estimand for DiD is the average treatment effect on the treated (ATT) $\E[Y(1) - Y(0)\mid W = 1]$. A natural definition of rank-ATT is
\begin{equation}
  \label{eq:rate_did}
  \beta_{\did}^{*} = \rate(F_{Y_1(1)\mid W=1}, F_{Y_1(0)\mid W=1}).
\end{equation}
The parallel trend assumption for original outcomes states that
\[\E[Y_1(0) - Y_0(0)\mid W = 1] = \E[Y_1(0) - Y_0(0)\mid W = 0].\]
This can be rewritten as
\[\E[Y_1(0)\mid W=1] - \E[Y_1(0)\mid W = 0] = \E[Y_0(0)\mid W=1] - \E[Y_0(0)\mid W = 0].\]
The left-hand side is the mean difference between $F_{Y_1(0)\mid W = 1}$ and $F_{Y_1(0)\mid W=0}$ and the right-hand side is the mean difference between $F_{Y_0(0)\mid W = 1}$ and $F_{Y_0(0)\mid W=0}$. Motivated by this expression, we define the rank parallel trend assumption as follows.
\begin{assumption}\label{as:rank_PT}
  $\rate(F_{Y_1(0)\mid W=1}, F_{Y_1(0)\mid W=0}) = \rate(F_{Y_0(0)\mid W=1}, F_{Y_0(0)\mid W=0})$.
\end{assumption}
Note that Assumption \ref{as:rank_PT} is not equivalent to
\begin{equation}
  \label{eq:alternative_rank_PT}
  \rate(F_{Y_1(0)\mid W=1}, F_{Y_0(0)\mid W=1}) = \rate(F_{Y_1(0)\mid W=0}, F_{Y_0(0)\mid W=0}).
\end{equation}
We show that the identifying assumption for CiC implies Assumption \ref{as:rank_PT} but not \eqref{eq:alternative_rank_PT}.
\begin{example}\label{ex:rank_preserve}
  Let $U_i$ be a unit-level time-invariant unobserved confounder $U_{i}$ and $f_0, f_1$ be two strictly increasing functions. Further, let 
  \[Y_{it}(0) = f_t(U_i), \quad i = 1,\ldots,n, t=0,1.\]
  This generalizes the classical two-way-fixed-effects model and is shown to be equivalent to the identifying assumption for CiC \citep{athey2006identification}:
  \begin{equation}\label{eq:CiC}
    F_{Y_1(0)\mid W = 1} = F_{Y_0(0)\mid W = 1}\circ F_{Y_0(0)\mid W = 0}^{-1} \circ F_{Y_1(0)\mid W = 0}
  \end{equation}
  Then the rank of $Y_{it}(0)$ among all control potential outcomes at time $t$ is invariant across time. Any definition of rank parallel trend should be met by this example. In fact, Assumption \ref{as:rank_PT} holds because, for $t = 0,1$,
  \[\rate(F_{Y_t(0)\mid W=1}, F_{Y_t(0)\mid W=0}) = \rate(F_{f_t(U)\mid W=1}, F_{f_t(U)\mid W=0}) = \rate(F_{U\mid W=1}, F_{U\mid W=0}),\]
  where the last step applies the invariance property of $\beta^{*}$. By contrast, \eqref{eq:alternative_rank_PT} does not hold unless $f_0 = f_1$, in which case the distribution of $Y_t(0)$ remains constant over time.
\end{example}

We prove that the effective estimand for the standard rank-DiD estimator is not the rank-ATT and could have a different sign. 
\begin{theorem}\label{thm:rank_did}
  Under Assumptions \ref{as:SUTVA_did}, \ref{as:continuous_Y_did}, and \ref{as:rank_PT}, as $n\rightarrow\infty$,
  \[\hat{\beta}_{\did}\stackrel{p}{\rightarrow}\rate(F_{Y_1(1)\mid W=1}, F_{Y_1(0)\mid W=0}) - \rate(F_{Y_1(0)\mid W=1}, F_{Y_1(0)\mid W=0}).\]
  Without further assumptions, the right-hand side is not equal to $\beta^{*}_{\did}$. Furthermore, replacing $\RY_{it}$ by their ranks among the treated or control units at time $t$ does not change the limit.
\end{theorem}

To identify the rank-LATE, we need to rank outcomes based on another reference distribution. Let $\hat{F}_{Y_1(0)\mid W=1}$ be an estimate of the counterfactual distribution $F_{Y_1(0)\mid W=1}$. We define a modified rank-DiD estimator by replacing $\RY_{i1} / n$ with $\hat{F}_{Y_1(0)\mid W=1}(Y_i)$:
\begin{equation}
  \label{eq:mdid}
  \hat{\beta}_{\mdid} = \frac{1}{n_1}\sum_{i: W_i=1}\lb\hat{F}_{Y_1(0)\mid W=1}(Y_i) - \frac{\RY_{i0}}{n}\rb - \frac{1}{n_0}\sum_{i: W_i=0}\lb \hat{F}_{Y_1(0)\mid W=1}(Y_i) - \frac{\RY_{i0}}{n}\rb.
\end{equation}

\begin{theorem}\label{thm:rank_mdid}
  Assume that $\hat{F}_{Y_1(0)\mid W=1}$ is a consistent estimate of $F_{Y_1(0)\mid W=1}$ in the sense of \eqref{eq:rank_approx}. In the setting of Theorem \ref{thm:rank_did}, as $n\rightarrow\infty$,
  \[\hat{\beta}_{\did}\stackrel{p}{\rightarrow}\beta^{*}_{\did}.\]
\end{theorem}

Unfortunately, the rank parallel trend assumption does not imply identification of $F_{Y_1(0)\mid W=1}$. Nevertheless, it can be consistently estimated under the identifying assumption of CiC \citep{athey2006identification} or distributional DiD \citep{roth2023parallel}. Since the former implies the rank parallel trend condition, the rank-ATT can be identified under the same assumptions as CiC. 

\subsection{Rank-RDD}
We consider the standard RDD setting where $X_i$ is a continuous one-dimensional running variable and $x^*$ is the common cutoff for all units. For the purpose of  exposition, we focus on sharp RDDs where 
\begin{equation}\label{eq:sharp_RDD}
W_i = I(X_i\ge x^*).
\end{equation}
The extensions to fuzzy RDDs is straightforward yet mathematically involved. Let $(Y_i(1), Y_i(0))$ denote the potential outcomes. Following \cite{imbens2008regression}, we make the distributional continuity assumption.
\begin{assumption}\label{as:RDD}
(Continuity of conditional distribution function) for $w\in \{0,1\}$, $F_{Y(w)\mid X=x}$ weakly converges to $F_{Y(w)\mid X=x^*}$ as $x\rightarrow x^*$.
\end{assumption}
As \cite{imbens2008regression} noted, for level outcomes, the continuity assumption is stronger than required for identification but practically indistinguishable with the weakest possible assumption in most cases. In addition, when the outcomes are ranked, we need this stronger continuity assumption.

While there are many commonly-used RDD estimators for level outcomes \citep[e.g.][]{imbens2008regression,lee2010regression,calonico2014robust,imbens2019optimized,eckles2020noise,cattaneo2022regression}, we focus on the simplest kernel estimator
\begin{equation}\label{eq:rank_RDD}
\hat{\beta}_{\rdd} = \frac{\sum_{i: X_i \ge x^*}\frac{\RY_i}{n} K\lb\frac{X_i - x^*}{h_n}\rb}{\sum_{i: X_i \ge x^*} K\lb\frac{X_i - x^*}{h_n}\rb} - \frac{\sum_{i: X_i < x^*}\frac{\RY_i}{n} K\lb\frac{X_i - x^*}{h_n}\rb}{\sum_{i: X_i < x^*} K\lb\frac{X_i - x^*}{h_n}\rb}.
\end{equation}
Further, let $\hat{\beta}_{\rdd, 1}$ and $\hat{\beta}_{\rdd, 0}$ denote the same estimator but with $\RY_i/n$ replaced by $\RY_{i,1}/n_1$ and $\RY_{i,0}/n_0$, respectively. Here we make the following standard assumptions on the density of $X$ and the kernel. 

\begin{assumption}\label{as:kernel}
  \begin{enumerate}[(a)]
  \item $X$ has a Lipschitz continuous density $p$ (with respect to the Lebesgue measure) with $x^{*}$ in the interior of the support and $p(x^{*}) \in (0, \infty)$.
  \item $K$ has a bounded support with 
  \[ \int K^2(u)du < \infty, \quad \int_{0}^\infty K(u)du \cdot \int_{-\infty}^0 K(u)du \neq 0.\]
  \end{enumerate}
\end{assumption}

The following result shows that the kernel RDD estimator does not satisfy the reference robustness and none of these estimators converge to the rank-ATE at cutoff $x^*$: 
\begin{equation}\label{eq:cutoff_rank_LATE}
\rate(F_{Y(1)\mid x^*}, F_{Y(0)\mid x^*}),
\end{equation}
where $F_{Y(w)\mid x}$ denotes the conditional distribution of $Y(w)$ given $X = x$.

\begin{theorem}\label{thm:rank_rdd}
  Assume that $h_n\rightarrow 0$ and $nh_n\rightarrow \infty$. Under Assumptions \ref{as:SUTVA}, \ref{as:continuous_Y}, \ref{as:RDD}, and \ref{as:kernel}, as $n\rightarrow \infty$,
    \[\hat{\beta}_{\rdd}\stackrel{p}{\rightarrow} \rate(F_{Y(1)\mid x^*}, F_{Y}) - \rate(F_{Y(0)\mid x^*}, F_{Y}),\]
    \[\hat{\beta}_{\rdd,1}\stackrel{p}{\rightarrow} \rate(F_{Y(1)\mid x^*}, F_{Y\mid W=1}) - \rate(F_{Y(0)\mid x^*}, F_{Y\mid W=1}),\]
    \[\hat{\beta}_{\rdd,0}\stackrel{p}{\rightarrow} \rate(F_{Y(1)\mid x^*}, F_{Y\mid W=0}) - \rate(F_{Y(0)\mid x^*}, F_{Y\mid W=0}),\]
\end{theorem}

Since we only focus on identification, our results would continue to hold for local-polynomial estimators \citep[e.g.][]{imbens2008regression,calonico2014robust} using the standard asymptotic theory of local-polynomial regression estimators \citep{fan1992variable}.

Intuitively, to identify rank-ATE at the cutoff, we need to rank based on units near the cutoff. This motivates the following kernel U-statistic: 
\begin{equation}\label{eq:kernel_U}
\hat{\beta}_{\mrdd} = \frac{\sum_{i: X_i \ge x^*}\sum_{j: X_j < x^*} K\lb\frac{X_i - x^*}{h_n}\rb K\lb\frac{X_j - x^*}{h_n}\rb I(Y_j \le Y_i)}{\sum_{i: X_i \ge x^*}\sum_{j: X_j < x^*}K\lb\frac{X_i - x^*}{h_n}\rb K\lb\frac{X_j - x^*}{h_n}\rb }
\end{equation}
Note that \eqref{eq:kernel_U} is equivalent to \eqref{eq:rank_RDD} with $\RY_i/n$ replaced by a weighted rank
\[\frac{\sum_{j: X_j < x^*} K\lb\frac{X_j - x^*}{h_n}\rb I(Y_j \le Y_i)}{\sum_{j: X_j < x^*}K\lb\frac{X_j - x^*}{h_n}\rb }.\]
\begin{theorem}\label{thm:rank_mrdd}
  Assume that $h_n\rightarrow 0$ and $nh_n\rightarrow \infty$. Under Assumptions \ref{as:SUTVA}, \ref{as:continuous_Y} and \ref{as:RDD}, as $n\rightarrow \infty$,
    \[\hat{\beta}_{\mrdd}\stackrel{p}{\rightarrow} \rate(F_{Y(1)\mid x^{*}}, F_{Y(0)\mid x^{*}}).\]
\end{theorem}

\section{Summary and discussions}
In this paper, we study the effective estimands and causal interpretations of popular econometric methods with ranked outcomes or treatments. We introduce the rank-ATE as a primitive causal parameter that enjoys many desirable properties and serves as the building block for other estimands studied in this paper. It also allows us to generalize the estimands for 2SLS, DiD, and RDD to the setting with ranks, though we show that they are not identified by directly applying these methods to ranked outcomes due to the nonlinearity of rank-ATE. We develop alternative identification strategies based on different reference distributions for ranking. 

\subsection{An alternative primitive estimand}
  While the rank-ATE is a measure of the comparison between the distributions of treated and control potential outcomes, it cannot be written as an average of individual-level treatment effects, i.e., $\E[g(Y(1), Y(0))]$ \citep{chen2023logs}. An analogue of rank-ATE is
  \begin{equation}
    \label{eq:taustar}
    \rate^{*}(F_{Y(1)}, F_{Y(0)}) = \P(Y(1)\ge Y(0)) - \frac{1}{2} = \E\left[I(Y(1)\ge Y(0)) - \frac{1}{2}\right]. 
  \end{equation}
  It is not hard to find data generating distributions under which $\rate(F_{Y(1)}, F_{Y(0)})$ and $\rate^{*}(F_{Y(1)}, F_{Y(0)})$ have different signs. This is also referred to as the Hand's paradox in biostatistics \citep[e.g.][]{fay2018causal}.

  Apparently, $\rate^{*}(F_{Y(1)}, F_{Y(0)})$ is only partially identifiable without further assumptions. \cite{fan2010sharp} derive the tight lower and upper bounds on $\tau^{*}$:
  \[\tau_{L}^{*}(F_{Y(1)}, F_{Y(0)}) = \max\left\{\sup_{y}(F_{Y(0)}(y) - F_{Y(1)}(y)), 0\right\} - \frac{1}{2}, \]
  and
  \[\tau_{R}^{*}(F_{Y(1)}, F_{Y(0)}) = \min\left\{\inf_{y}(F_{Y(0)}(y) - F_{Y(1)}(y)), 0\right\} + \frac{1}{2}.\]
  With covariates, the bounds can be sharpened \citep{lee2021partial}.

  However, to consistently estimate and quantify uncertainty of the sharper bounds, \cite{lee2021partial} requires a correct model for the conditional distributions of potential outcomes given the covariates $X$, which is arguably too strong when $X$ is continuous, mixed-type, high-dimensional, or unstructured (e.g., texts). \cite{ji2023model} overcomes this issue by exploiting Kantorovich duality in optimal transport theory. In particular, our method can wrap around any estimates of $\P(Y(1)\mid X)$ and $\P(Y(0)\mid X)$ and generate valid lower and upper bounds for $\tau^{*}$ even if the estimates are completely off. In the meantime, our bounds are efficient if the estimates of conditional distributions are consistent with the semiparametric rates (i.e., $O(n^{-1/4})$). 

\bibliography{rank_reg_causal}
\bibliographystyle{plainnat}

  \appendix

  \section{Proofs}\label{app:proofs}
  \subsection{Proofs for OLS}
    \begin{proof}[\textbf{Proof of Theorem \ref{thm:binary_noz_randomized}}]
      It is easy to see that the estimator is equivalent to the difference-in-means estimator:
      \begin{equation}\label{eq:hatbeta_noz}
        \hat{\beta}_{\nox} = \frac{1}{n_1}\sum_{i: W_i=1}\frac{\RY_i}{n} - \frac{1}{n_0}\sum_{i: W_i = 0}\frac{\RY_i}{n}.
        \end{equation}
    By \eqref{eq:rank_approx},
    \[\hat{\beta}_{\nox} = \hat{\beta}_{\nox}^{*} + O_\P(1/\sqrt{n}) \quad \text{where } \hat{\beta}_{\nox}^{*} = \frac{1}{n_1}\sum_{i: W_i=1}F_Y(Y_i) - \frac{1}{n_0}\sum_{i: W_i = 0}F_Y(Y_i).\]
    We can rewrite $\hat{\beta}_{\nox}^{*}$ as
    \[\hat{\beta}_{\nox}^{*} = \frac{1}{n}\sum_{i=1}^{n}\lb \frac{F_Y(Y_i)W_i}{n_1/n} - \frac{F_Y(Y_i)(1-W_i)}{n_0/n}\rb.\]
    Note that $n_1/n \stackrel{p}{\rightarrow}\pi$ and $n_0/n \stackrel{p}{\rightarrow}1 - \pi$. By Law of Large Number,
    \begin{equation}\label{eq:limit_betahat}
      \hat{\beta}_{\nox}^{*}\stackrel{p}{\rightarrow}\frac{\E[F_Y(Y)W]}{\pi} - \frac{\E[F_Y(Y)(1-W)]}{1-\pi} = \frac{\E[F_Y(Y(1))W]}{\pi} - \frac{\E[F_Y(Y(0))(1-W)]}{1-\pi}.
      \end{equation}
    where the last step invokes SUTVA. Under Assumption \ref{as:randomized},
    \[\E[F_Y(Y(1))W] = \pi\E[F_Y(Y(1))], \quad \E[F_Y(Y(0))(1-W)] = (1-\pi)\E[F_Y(Y(0))].\]
    Thus,
    \[\hat{\beta}_{\nox}^{*}\stackrel{p}{\rightarrow} \E[F_Y(Y(1))] - \E[F_Y(Y(0))].\]
    By \eqref{eq:mixture},
    \[\E[F_Y(Y(1))] = \pi \E[F_{Y(1)}(Y(1))] + (1 - \pi)\E[F_{Y(0)}(Y(1))] = \frac{\pi}{2} + (1 - \pi)\int F_{Y(0)}(y) dF_{Y(1)}(y),\]
    where the last step uses the fact that $F_{Y(1)}(Y(1))\sim \mathrm{Unif}([0,1])$ when $F_{Y(1)}$ is continuous. Similarly,
    \[\E[F_Y(Y(0))] =  \frac{1-\pi}{2} + \pi\int F_{Y(1)}(y) dF_{Y(0)}(y).\]
    Using integration by part, we have
    \[\int F_{Y(1)}(y) dF_{Y(0)}(y) = 1 - \int F_{Y(0)}(y) dF_{Y(1)}(y).\]
    By \eqref{eq:rate_alternative}, 
    \[\E[F_Y(Y(1))] - \E[F_Y(Y(0))] = \rate(F_{Y(1)}, F_{Y(0)}).\]
  \end{proof}

  \begin{proof}[\textbf{Proof of Theorem \ref{thm:reference_robustness}}]
    Similar to \eqref{eq:hatbeta_noz}, we have
    \[\hat{\beta}_{\nox, 1} = \frac{1}{n_1}\sum_{i}W_i\frac{\RY_{i,1}}{n_1} - \frac{1}{n_0}\sum_{i}(1 - W_i)\frac{\RY_{i,1}}{n_1},\]
    and
    \[\hat{\beta}_{\nox, 0} = \frac{1}{n_1}\sum_{i}W_i\frac{\RY_{i,0}}{n_0} - \frac{1}{n_0}\sum_{i}(1 - W_i)\frac{\RY_{i,0}}{n_0}.\]
    By definition of $\RY_{i,1}$, we can rewrite it as
    \[\hat{\beta}_{\nox, 1} = \frac{1}{n_1^2}\sum_{i,j}W_iW_j I(Y_j\le Y_i) - \frac{1}{n_1n_0}\sum_{i,j}(1-W_i)W_j I(Y_j\le Y_i).\]
    Similarly,
    \[\hat{\beta}_{\nox, 0} = \frac{1}{n_1n_0}\sum_{i,j}W_i(1-W_j) I(Y_j\le Y_i) - \frac{1}{n_0^2}\sum_{i,j}(1-W_i)(1-W_j) I(Y_j\le Y_i).\]
    In the absence of ties, $I(Y_j\le Y_i) + I(Y_i \le Y_j) = 1$ for any $i\neq j$. Swapping the labels $i$ and $j$,
    \begin{align*}
      &\frac{1}{n_1^2}\sum_{i,j}W_iW_j I(Y_j\le Y_i) = \frac{1}{2n_1^2}\sum_{i\neq j}W_iW_j + \frac{1}{n_1^2}\sum_{i}W_i\\
      & = \frac{1}{2n_1^2}\left\{\lb \sum_{i}W_i\rb^2 - \sum_{i}W_i^2\right\} + \frac{1}{n_1^2}\sum_{i}W_i = \frac{1}{2n_1^2} (n_1^2 - n_1) + \frac{1}{n_1} = \frac{1}{2} + \frac{1}{2n_1}.
    \end{align*}
    Similarly,
    \[\frac{1}{n_0^2}\sum_{i,j}(1-W_i)(1-W_j)I(Y_j\le Y_i) = \frac{1}{2} + \frac{1}{2n_0}.\]
    In addition,
    \begin{align*}
      & \frac{1}{n_1n_0}\sum_{i,j}(1-W_i)W_j I(Y_j\le Y_i) + \frac{1}{n_1n_0}\sum_{i,j}W_i(1-W_j) I(Y_j\le Y_i)\\
      & = \frac{1}{n_1n_0}\sum_{i,j}(1-W_i)W_j I(Y_j\le Y_i) + \frac{1}{n_1n_0}\sum_{i,j}(1-W_i)W_j I(Y_i\le Y_j)\\
      & = \frac{1}{n_1n_0}\sum_{i,j}(1-W_i)W_j = 1.
    \end{align*}
    Let
    \[\td{\beta} = \frac{1}{n_0}\sum_{i}(1 - W_i)\frac{\RY_{i,0}}{n}.\]
    Then $\hat{\beta}_{\nox, 1}$ and $\hat{\beta}_{\nox, 0}$ can be rewritten as
    \[\hat{\beta}_{\nox, 1} = \frac{1}{2} - \td{\beta} + \frac{1}{2n_1}, \quad \hat{\beta}_{\nox, 1} = \frac{1}{2} - \td{\beta} - \frac{1}{2n_0}.\]
    Noting that
    \[\frac{\RY_i}{n} = \frac{n_1}{n}\frac{\RY_{i,1}}{n_1} + \frac{n_0}{n}\frac{\RY_{i,0}}{n_0},\]
    we have
    \[\hat{\beta}_{\nox} = \frac{n_1}{n}\hat{\beta}_{\nox, 1}+\frac{n_0}{n}\hat{\beta}_{\nox, 0} = \frac{1}{2} - \td{\beta}.\]
    The proof is then completed.
    \end{proof}

  \begin{proof}[\textbf{Proof of Theorem \ref{thm:binary_z_randomized}}]
    Under Assumption \ref{as:randomized}, 
    \[\Cov(W, X) = 0, \quad \Cov(W, W(X - \E[X])) = \E[W(X - \E[X])] - \E[W]\E[W(X - \E[X])] = 0.\]
    The population Frisch-Waugh-Lovell theorem implies that the limits of $\hat{\beta}$ and $\hat{\beta}_{\inter}$ both coincide with the limit of $\hat{\beta}_{\nox}$. 
  \end{proof}

\begin{proof}[\textbf{Proof of Theorem \ref{thm:weighted_average}}]
The population Frisch-Waugh-Lovell theorem and \eqref{eq:rank_approx} together imply that 
\[\hat{\beta} \stackrel{p}{\rightarrow} \frac{\E[(W - \td{\pi}(X))F_Y(Y)]}{\E[(W - \td{\pi}(X))^2]}.\]
The numerator can be decomposed into 
\[\E[(W - \td{\pi}(X))F_Y(Y(0))] + \E[(W - \td{\pi}(X))W(F_Y(Y(1)) - F_Y(Y(0))].\]
By the law of iterated expectation and Assumption \ref{as:strong_ignorability}, the first term can be expressed as 
\[\E[(\pi(X) - \td{\pi}(X))\E[F_Y(Y(0))\mid X]].\]
When condition (a) holds, it is clearly zero; when condition (b) holds, 
\[\E[(\pi(X) - \td{\pi}(X))\E[F_Y(Y(0))\mid X]] = \E[(\pi(X) - \td{\pi}(X))X^\top \eta^{*}] = 0,\]
because $\td{\pi}(X)$ is the linear projection of $\pi(X)$ and $\E[(\pi(X) - \td{\pi}(X))X] = 0$. Thus, the first term is zero under either condition (a) or (b) and the numerator can be simplified into 
\begin{align*}
& \E[(W - \td{\pi}(X))W(F_Y(Y(1)) - F_Y(Y(0))]\\
& = \E[\pi(X) (1 - \td{\pi}(X))\E[F_Y(Y(1)) - F_Y(Y(0))\mid X]]\\
& = \E[w(X) \rate(F_{Y(1)\mid X}, F_{Y(0)\mid X})].
\end{align*}
Now we turn to the denominator. Since $\td{\pi}(X)$ is the linear projection of $\pi(X)$, 
\[\E[\td{\pi}(X)(\pi(X) - \td{\pi}(X))] = 0.\]
Thus, 
\[\E[(W - \td{\pi}(X))^2] = \E[\pi(X) - 2\pi(X)\td{\pi}(X) + \td{\pi}^2(X)] = \E[w(X)].\]
The proof is then completed.
\end{proof}

    \begin{proof}[\textbf{Proof of Theorem \ref{thm:general_noz_randomized}}]
      Under Assumption \ref{as:randomized},
      \[\Cov(h_n(W), X) = 0.\]
      By population Frisch-Waugh-Lovell theorem, \eqref{eq:rank_approx}, and Assumption \ref{as:limit_h}, we can show that
      \begin{equation}
        \hat{\beta}_{h_n}\stackrel{p}{\rightarrow} \frac{\Cov(F_Y(Y), h(W))}{\Var[h(W)]}.
        \end{equation}
    By Assumption \ref{as:continuous_Y_cont}, $F_Y$ is continuous. Thus, $\E[F_Y(Y)] = 1/2$ and
    \begin{align*}
      &\Cov(F_Y(Y), h(W)) = \E\left[\left\{F_Y(Y) - \frac{1}{2}\right\}h(W)\right]\\
      & = \int \E\left[F_Y(Y) - \frac{1}{2} \mid W=w\right] h(w)dF_{W}(w)\\
      & = \int \E\left[F_Y(Y(w)) - \frac{1}{2}\mid W=w\right] h(w)dF_{W}(w) \quad \text{(by Assumption \ref{as:SUTVA_cont})}\\
      & = \int \E\left[F_Y(Y(w)) - \frac{1}{2}\right] h(w)dF_{W}(w) \quad \text{(by Assumption \ref{as:randomized_cont})}\\
      & = \int \rate(F_{Y(w)}, F_Y)h(w)dF_{W}(w) \quad \text{(by \eqref{eq:rate_alternative})}\\
      & = \int \rate(F_{Y(w)}, F_{Y(\td{w})})h(w)dF_{W}(w)dF_{W}(\td{w}) \quad \text{(by linearity)}\\
      & = \int_{w\neq \td{w}} \rate(F_{Y(w)}, F_{Y(\td{w})})h(w)dF_{W}(w)dF_{W}(\td{w}) \quad \text{(by Assumption \ref{as:continuous_Y_cont} and anti-symmetry)}\\
      & = \int_{w > \td{w}} \rate(F_{Y(w)}, F_{Y(\td{w})})(h(w) - h(\td{w}))dF_{W}(w)dF_{W}(\td{w})\quad \text{(by Assumption \ref{as:continuous_Y_cont} and anti-symmetry)}\\
      & = \E\left[b(W, \td{W})\rate\lb F_{Y(W)}, F_{Y(\td{W})}\rb\right].
    \end{align*}
    The proof is completed by applying the well-known variance formula $\Var[X] = (1/2)\E[(X - X')^2] = \E[(X - X')^2 I(X > X')]$ with $X = h(W)$.
  \end{proof}

  \subsection{Proofs for 2SLS}

  \begin{proof}[\textbf{Proof of Theorem \ref{thm:2sls}}]
    Let $m_1$ and $m_0$ denote the number of units with $Z_i=1$ and $Z_i=0$, respectively. The 2SLS estimator can be expressed as 
    \[\hat{\beta}_{\sls} = \frac{\frac{1}{m_1}\sum_{i: Z_i=1}\frac{\RY_i}{n} - \frac{1}{m_0}\sum_{i: Z_i=0}\frac{\RY_i}{n}}{\frac{1}{m_1}\sum_{i: Z_i=1}W_i - \frac{1}{m_0}\sum_{i: Z_i=0}W_i}.\]
    By Assumption \ref{as:IV} (d), 
    \[\frac{1}{m_1}\sum_{i: Z_i=1}W_i - \frac{1}{m_0}\sum_{i: Z_i=0}W_i\stackrel{p}{\rightarrow} \P(W=1\mid Z=1) - \P(W=1\mid Z=0) > 0.\]
    By \eqref{eq:rank_approx}, 
    \begin{align*}
    & \frac{1}{m_1}\sum_{i: Z_i=1}\frac{\RY_i}{n} - \frac{1}{m_0}\sum_{i: Z_i=0}\frac{\RY_i}{n}\\
    & = \frac{1}{m_1}\sum_{i: Z_i=1}F_Y(Y_i) - \frac{1}{m_0}\sum_{i: Z_i=0}F_Y(Y_i) + o_\P(1).
    \end{align*}
    Thus, 
    \[\hat{\beta}_{\sls} = \hat{\beta}_{\sls}^{*} + o_\P(1),\]
    where $\hat{\beta}_{\sls}^{*}$ is the 2SLS estimator with $F_{Y}(Y_i)$s being outcomes. The standard theory of 2SLS \citep{imbens1994identification} implies that 
    \[\hat{\beta}_{\sls}^{*}\stackrel{p}{\rightarrow} \E[F_Y(Y(1)) - F_Y(Y(0))\mid G = \tc]. \]
    By \eqref{eq:rate_alternative},
    \[\E[F_Y(Y(1))\mid G = \tc] = \rate(F_{Y(1)\mid \tc}, F_Y),\]
    and 
    \[\E[F_Y(Y(0))\mid G = \tc] = \rate(F_{Y(0)\mid \tc}, F_Y).\]    
    The proof for $\hat{\beta}_{\sls}$ is then completed. The limits of other estimators can be proved similarly. 
  \end{proof}

  \begin{proof}[\textbf{Proof of Theorem \ref{thm:complier_ranking}}]
  Following the same steps in the proof of Theorem \ref{thm:2sls}, we can show that 
  \[\hat{\beta}_{\sls, \zeta, \tc} \stackrel{p}{\rightarrow} \rate(F_{Y(1)\mid \tc}, \zeta \hat{F}_{Y(1)\mid \tc} + (1 - \zeta)\hat{F}_{Y(0)\mid \tc}) - \rate(F_{Y(1)\mid \tc}, \zeta \hat{F}_{Y(1)\mid \tc} + (1 - \zeta)\hat{F}_{Y(0)\mid \tc}).\]
  By reference robustness, the RHS equals the rank-LATE.
  \end{proof}
  
 \begin{proof}[\textbf{Proof of Theorem \ref{thm:complier_ranking_estimator}}]
    By Assumptions \ref{as:SUTVA} and \ref{as:IV},
  \begin{align*}
    &F_{11}(y) = \P(Y(1)\le y\mid W(1) = 1, Z = 1)\\
     & =\P(Y(1)\le y\mid W(1) = 1)\\
     & = \P(Y(1)\le y\mid G\in \{\ta, \tc\})\\
    & = \P(G = \ta\mid G\in \{\ta, \tc\})F_{Y(1)\mid \ta}(y) + \P(G = \tc\mid G\in \{\ta, \tc\})F_{Y(1)\mid \tc}(y)\\
    & = \frac{\pi_{\ta}}{\pi_{\ta} + \pi_{\tc}}F_{Y(1)\mid \ta}(y) + \frac{\pi_{\tc}}{\pi_{\ta} + \pi_{\tc}}F_{Y(1)\mid \tc}(y),
  \end{align*}
  and
  \begin{align*}
    &F_{10}(y) = \P(Y(1)\le y\mid W(0) = 1, Z = 0)\\
    & = \P(Y(1)\le y\mid W(0) = 1)\\
    & = \P(Y(1)\le y\mid G = \ta)\\
    & = F_{Y(1)\mid \ta}(y).
  \end{align*}
   As a result,
  \begin{equation}
    \label{eq:F1c}
    F_{Y(1)\mid \tc} = \frac{\pi_{\ta} + \pi_{\tc}}{\pi_{\tc}}\lb F_{11} - \frac{\pi_{\ta}}{\pi_{\ta} + \pi_{\tc}}F_{10}\rb.
  \end{equation}
  Similarly,
  \begin{equation}
    \label{eq:F0c}
    F_{Y(0)\mid \tc} = \frac{\pi_{\tn} + \pi_{\tc}}{\pi_{\tc}}\lb F_{00} - \frac{\pi_{\tn}}{\pi_{\tn} + \pi_{\tc}}F_{01}\rb.
  \end{equation}
  \end{proof}

  \subsection{Proofs for DiD}
  \begin{proof}[\textbf{Proof of Theorem \ref{thm:rank_did}}]
  By \eqref{eq:rank_approx} and SUTVA,
  \begin{align}
    \hat{\beta}_{\did}&\stackrel{p}{\rightarrow} \E[F_{Y_1}(Y_1)\mid W=1] - \E[F_{Y_1}(Y_1)\mid W=0] - (\E[F_{Y_0}(Y_0)\mid W=1] - \E[F_{Y_0}(Y_0)\mid W=0]).\label{eq:did_limit}
  \end{align}
  Since $F_{Y_1} = \pi F_{Y_1\mid W=1} + (1 - \pi)F_{Y_1\mid W=0}$ and $F_{Y_0} = \pi F_{Y_0\mid W=1} + (1 - \pi)F_{Y_0\mid W=0}$, by partial additivity and \eqref{eq:rate_alternative}, 
  \begin{align}
    &\E[F_{Y_1}(Y_1)\mid W=1] - \E[F_{Y_1}(Y_1)\mid W=0] = \rate(F_{Y_1\mid W=1}, F_{Y_1}) - \rate(F_{Y_1\mid W=0}, F_{Y_1})\nonumber\\
    & = \rate(F_{Y_1\mid W=1}, F_{Y_1\mid W=0}),   \label{eq:first_diff}
  \end{align}
  and
  \begin{align}
    &\E[F_{Y_0}(Y_0)\mid W=1] - \E[F_{Y_0}(Y_0)\mid W=0] = \rate(F_{Y_0\mid W=1}, F_{Y_0}) - \rate(F_{Y_0\mid W=0}, F_{Y_0})\nonumber\\
    & = \rate(F_{Y_0\mid W=1}, F_{Y_0\mid W=0}). \label{eq:second_diff}
    \end{align}
  By Assumption \ref{as:rank_PT},
  \begin{align*}
    &\rate(F_{Y_0\mid W=1}, F_{Y_0\mid W=0}) = \rate(F_{Y_0(0)\mid W=1}, F_{Y_0(0)\mid W=0}) = \rate(F_{Y_1(0)\mid W=1}, F_{Y_1(0)\mid W=0})
  \end{align*}
  Combining above pieces, we have
  \begin{align*}
    \hat{\beta}_{\did}&\stackrel{p}{\rightarrow} \rate(F_{Y_1\mid W=1}, F_{Y_1\mid W=0}) - \rate(F_{Y_1(0)\mid W=1}, F_{Y_1(0)\mid W=0})\\
    & = \rate(F_{Y_1(1)\mid W=1}, F_{Y_1(0)\mid W=0}) - \rate(F_{Y_1(0)\mid W=1}, F_{Y_1(0)\mid W=0}).
  \end{align*}
  Finally, using other ranks would not change \eqref{eq:first_diff} or \eqref{eq:second_diff} due to partial additivity. 
\end{proof}

\begin{proof}[\textbf{Proof of Theorem \ref{thm:rank_mdid}}]
  Similar to \eqref{eq:did_limit} in the proof of Theorem \ref{thm:rank_did},
  \begin{align*}
    \hat{\beta}_{\mdid}&\stackrel{p}{\rightarrow} \E[F_{Y_1(0)\mid W = 1}(Y_1)\mid W=1] - \E[F_{Y_1(0)\mid W=1}(Y_1)\mid W=0]\\
                       & \quad - (\E[F_{Y_0}(Y_0)\mid W=1] - \E[F_{Y_0}(Y_0)\mid W=0])\\
                       & = \beta^{*}_{\did} + \lb\frac{1}{2} - \E[F_{Y_1(0)\mid W=1}(Y_1(0))\mid W=0]\rb\\
                       & \quad - (\E[F_{Y_0}(Y_0)\mid W=1] - \E[F_{Y_0}(Y_0)\mid W=0])\\
                       & = \beta^{*}_{\did} + \rate(F_{Y_1(0)\mid W=1}, F_{Y_1(0)\mid W=0}) - \rate(F_{Y_0(0)\mid W=1}, F_{Y_0(0)\mid W=0})\\
    & = \beta^{*}_{\did},
  \end{align*}
  where the last step is due to Assumption \ref{as:rank_PT}.
\end{proof}

  \subsection{Proofs for RDD}
  We start by proving a lemma.
\begin{lemma}\label{lem:rank_rdd}
In the setting of Theorem \ref{thm:rank_rdd}, for any bounded random variable $Z_i$ such that $\E[Z \mid X=x]$ is continuous at $x=x^*$, as $n\rightarrow \infty$,
\[\frac{1}{n}\sum_{i:X_i\ge x^*}Z_i\cdot \frac{1}{h_n}K\lb\frac{X_i - x^*}{h_n}\rb = \E[Z \mid X=x^*]\cdot p(x^*) \int_{0}^\infty K(u)du + o_\P(1),\]
and 
\[\frac{1}{n}\sum_{i:X_i < x^*}Z_i\cdot \frac{1}{h_n}K\lb\frac{X_i - x^*}{h_n}\rb = \E[Z\mid X=x^*]\cdot p(x^*) \int_{-\infty}^0 K(u)du + o_\P(1).\]
\end{lemma}
\begin{proof}
Without loss of generality, assume $K$ is supported on $[-1,1]$ and bounded by $B$ and $p$ is $L$-Lipschitz. Note that
\begin{align*}
&\E\left[Z I(X\ge x^*)\cdot \frac{1}{h_n}K\lb\frac{X - x^*}{h_n}\rb\right]\\
& = \E\left[Z I(X \in [x^*, x^* + h_n])\cdot \frac{1}{h_n}K\lb\frac{X - x^*}{h_n}\rb\right]\\
& = \E\left[\E[Z\mid X] I(X \in [x^*, x^* + h_n])\cdot \frac{1}{h_n}K\lb\frac{X - x^*}{h_n}\rb\right]\\
& = \E[Z\mid X=x^*]\cdot \E\left[I(X \in [x^*, x^* + h_n])\cdot \frac{1}{h_n}K\lb\frac{X - x^*}{h_n}\rb\right] \\
& \quad + \E\left[\lb \E[Z\mid X] - \E[Z\mid X=x^*]\rb I(X \in [x^*, x^* + h_n])\cdot \frac{1}{h_n}K\lb\frac{X - x^*}{h_n}\rb\right] 
\end{align*}
Since $\E[Z\mid X = x]$ is continuous at $x = x^{*}$ and $h_n\rightarrow 0$, for any $\epsilon > 0$, there exists a sufficiently large $n(\eps)$ such that, for any $n\ge n(\eps)$, $|\E[Z\mid X] - \E[Z\mid X=x^*]|\le \eps$. Thus,
\begin{align*}
&\E\left[Z I(X\ge x^*)\cdot \frac{1}{h_n}K\lb\frac{X - x^*}{h_n}\rb\right]\\
& = (\E[Z\mid X=x^*] + o(1))\cdot \E\left[I(X \in [x^*, x^* + h_n])\cdot \frac{1}{h_n}K\lb\frac{X - x^*}{h_n}\rb\right].
\end{align*}
Since $X$ is supported on $[-1, 1]$, 
\[\E\left[I(X \in [x^*, x^* + h_n])\cdot \frac{1}{h_n}K\lb\frac{X - x^*}{h_n}\rb\right] = \int_{0}^1 K(u)p(x^* + h_n u)du.\]
By Cauchy-Schwarz inequality, 
\[\int_{-1}^{1} u |K(u)|du \le \sqrt{\int_{-1}^{1}u^2du \int_{-1}^{1}K^2(u)du} < \infty.\]
Since $p$ is $L$-Lipschitz,
\[\Bigg|\int_{0}^1 K(u)(p(x^* + h_n u) - p(x^*))du \Bigg|\le Lh_n \int_{0}^1 u|K(u)|du =O(h_n) = o(1).\]
Thus, 
\[\E\left[Z I(X\ge x^*)\cdot \frac{1}{h_n}K\lb\frac{X - x^*}{h_n}\rb\right] = p(x^*)\int_0^1 K(u)du + o(1).\]
Since $\E[Z\mid X=x^{*}]$ and $p(x^*)$ are bounded, we have 
\begin{equation}\label{eq:mean_kernel}
\E\left[Z I(X\ge x^*)\cdot \frac{1}{h_n}K\lb\frac{X - x^*}{h_n}\rb\right] = \E[Z\mid X=x^*]\cdot p(x^*)\int_0^1 K(u)du + o(1).
\end{equation}
Now we compute the variance. Assume $|Z|\le C$. Then,
\begin{align*}
&\Var\left[Z I(X\ge x^*)\cdot \frac{1}{h_n}K\lb\frac{X - x^*}{h_n}\rb\right]\\
& \le \E\left[Z^2 I(X \in [x^*, x^* + h_n])\cdot \frac{1}{h_n^2}K^2\lb\frac{X - x^*}{h_n}\rb\right]\\
& \le C^2 \frac{1}{h_n}\int_{0}^{1}K^2(u)p(x^* + h_n u)du.
\end{align*}
Since $p$ is $L$-Lipschitz,
\[\Bigg|\int_{0}^{1}K^2(u)p(x^* + h_n u)du - \int_{0}^{1}K^2(u)p(x^*)du\Bigg|\le Lh_n \int_{0}^1 uK^2(u)du = O(h_n).\]
Thus, 
\begin{align}
&\Var\left[Z I(X\ge x^*)\cdot \frac{1}{h_n}K\lb\frac{X - x^*}{h_n}\rb\right]\le C^2 \lb \frac{p(x^*)}{h_n}\int_{0}^1 K^2(u)du + O(1)\rb = O\lb \frac{1}{h_n}\rb.\label{eq:var_kernel}
\end{align}
Combining \eqref{eq:mean_kernel} and \eqref{eq:var_kernel}, by Chebyshev's inequality, we obtain that 
\[\frac{1}{n}\sum_{i:X_i\ge x^*}Z_i\cdot \frac{1}{h_n}K\lb\frac{X_i - x^*}{h_n}\rb = \E[Z \mid X=x^*]\cdot p(x^*) \int_{0}^1 K(u)du + o_\P(1) + O_\P\lb\frac{1}{\sqrt{n h_n}}\rb.\]
Since $nh_n \rightarrow \infty$, the last terms collapse to $o_\P(1)$. The integral range can be replaced by $[0, \infty)$ because $K$ is supported on $[-1, 1]$. The proof for the other result is similar. 
\end{proof}

\begin{proof}[\textbf{Proof of Theorem \ref{thm:rank_rdd}}]
Without loss of generality, assume $K$ is supported on $[-1,1]$ and bounded by $B$. Let 
\[D_{\ge} = \frac{1}{n}\sum_{i: X_i\ge x^*}\frac{1}{h_n}K\lb\frac{X_i - x^*}{h_n}\rb, \quad D_{<} = \frac{1}{n}\sum_{i: X_i < x^*}\frac{1}{h_n}K\lb\frac{X_i - x^*}{h_n}\rb.\]
By Lemma \ref{lem:rank_rdd} with $Z = 1$,
\begin{equation}\label{eq:bounded_phat}
D_{\ge} = p(x^{*})\int_{0}^1 K(u)du + o_\P(1), \quad D_{<} = p(x^{*})\int_{-1}^0 K(u)du + o_\P(1).    
\end{equation}
Define the terms in the numerator as 
\[N_{\ge} = \frac{1}{n}\sum_{i:X_i\ge x^*}\frac{\RY_i}{n} \cdot \frac{1}{h_n}K\lb\frac{X_i - x^*}{h_n}\rb, \quad N_{<} = \frac{1}{n}\sum_{i:X_i < x^*}\frac{\RY_i}{n} \cdot \frac{1}{h_n}K\lb\frac{X_i - x^*}{h_n}\rb.\]
By \eqref{eq:rank_approx},
\begin{equation}\label{eq:N_N*}
N_{\ge} = N_{\ge}^{*} + O_\P\lb\frac{1}{\sqrt{n}}\rb, \quad N_{<} = N_{<}^{*} + O_\P\lb\frac{1}{\sqrt{n}}\rb
\end{equation}
where 
\[N_{\ge}^{*} = \frac{1}{n}\sum_{i:X_i\ge x^*}F_Y(Y_i) \cdot \frac{1}{h_n}K\lb\frac{X_i - x^*}{h_n}\rb, \quad N_{<}^{*} = \frac{1}{n}\sum_{i:X_i < x^*}F_Y(Y_i) \cdot \frac{1}{h_n}K\lb\frac{X_i - x^*}{h_n}\rb.\]
By the SUTVA assumption (Assumption \ref{as:SUTVA}) and the sharp RDD assumption,
\begin{align*}
F_{Y\mid X=x} &= I(x\ge x^{*}) \cdot F_{Y(1)\mid X=x} + I(x < x^{*})\cdot F_{Y(0)\mid X=x}.
\end{align*}
By Assumption \ref{as:RDD}, 
\[F_{Y\mid X=x} \stackrel{d}{\rightarrow} F_{Y(0)\mid X=x^{*}}, \quad \text{as }x\uparrow x^*,\]
and 
\[F_{Y\mid X=x} \stackrel{d}{\rightarrow} F_{Y(1)\mid X=x^{*}}, \quad \text{as }x\downarrow x^*.\]
Since $F_Y$ is bounded, weak convergences and \eqref{eq:rate_alternative} imply
\begin{align*}
& \lim_{x\uparrow x^*}\E[F_Y(Y)\mid X = x] = \rate(F_{Y(0)\mid X=x^*}, F_{Y}).
\end{align*}
Similarly,
\begin{align*}
&\lim_{x\downarrow x^*}\E[F_Y(Y)\mid X = x] =\rate(F_{Y(1)\mid X=x^*}, F_{Y}).
\end{align*}
By Lemma \ref{lem:rank_rdd} with $Z = F_Y(Y)$, 
\begin{align*}
&N_{\ge}^{*} = \rate(F_{Y(1)\mid X=x^*}, F_{Y}) \cdot p(x^{*})\int_{0}^{1}K(u)du + o_\P(1),
\end{align*}
and 
\begin{align*}
&N_{<}^{*} = \rate(F_{Y(0)\mid X=x^*}, F_{Y}) \cdot p(x^{*})\int_{-1}^{0}K(u)du + o_\P(1).
\end{align*}
By \eqref{eq:N_N*}, 
\begin{align*}
&N_{\ge} = \rate(F_{Y(1)\mid X=x^*}, F_{Y}) \cdot p(x^{*})\int_{0}^{1}K(u)du + o_\P(1),
\end{align*}
and 
\begin{align*}
&N_{<} = \rate(F_{Y(0)\mid X=x^*}, F_{Y}) \cdot p(x^{*})\int_{-1}^{0}K(u)du + o_\P(1).
\end{align*}
Thus, 
\[\hat{\beta}_{\rdd} = \frac{N_{\ge}}{D_{\ge}} - \frac{N_{<}}{D_{<}} \stackrel{p}{\rightarrow} \rate(F_{Y(1)\mid X=x^*}, F_{Y}) - \rate(F_{Y(0)\mid X=x^*}, F_{Y}).\]
The limits for other ranking methods can be derived similarly. 
\end{proof}

\begin{proof}[\textbf{Proof of Theorem \ref{thm:rank_mrdd}}]
Similar to the proof of Theorem \ref{thm:rank_rdd}, we assume $K$ is 
 supported on $[-1, 1]$. Let 
\[D = \frac{1}{n^2}\sum_{i: X_i \ge x^*}\sum_{j: X_j < x^*}\frac{1}{h_n}K\lb\frac{X_i - x^*}{h_n}\rb \frac{1}{h_n}K\lb\frac{X_j - x^*}{h_n}\rb,\]
and 
\[N = \frac{1}{n^2}\sum_{i: X_i \ge x^*}\sum_{j: X_j < x^*}\frac{1}{h_n}K\lb\frac{X_i - x^*}{h_n}\rb \frac{1}{h_n}K\lb\frac{X_j - x^*}{h_n}\rb I(Y_j\le Y_i).\]
By Lemma \ref{lem:rank_rdd},
\begin{align*}
D &= \lb\sum_{i: X_i \ge x^*}\frac{1}{nh_n}K\lb\frac{X_i - x^*}{h_n}\rb \rb\lb \sum_{j: X_j < x^*}\frac{1}{nh_n}K\lb\frac{X_j - x^*}{h_n}\rb \rb\\
& = p^2(x^{*}) \lb \int_{0}^{1}K(u)du \rb \lb \int_{-1}^{0}K(u)du \rb + o_\P(1)
\end{align*}
It is left to prove 
\begin{equation}\label{eq:goal_mrdd}
N = p^2(x^{*}) \lb \int_{0}^{1}K(u)du \rb \lb \int_{-1}^{0}K(u)du \rb \cdot \rate(F_{Y(1)\mid x^*}, F_{Y(0)\mid x^*})+ o_\P(1).
\end{equation}
We rewrite $N$ as a U-statistic:
\[N = \frac{1}{n^2}\sum_{i\neq j}\frac{1}{h_n}K\lb\frac{X_i - x^*}{h_n}\rb I(X_i \ge x^*) \cdot \frac{1}{h_n}K\lb\frac{X_j - x^*}{h_n}\rb I(X_j < x^*) \cdot I(Y_j\le Y_i).\]
First we compute the mean of $N$. Since $(X_i, Y_i)$ are i.i.d., 
\begin{align*}
&\E\left[\frac{n}{n-1}N\right]\\
& = \E \left[\frac{1}{h_n}K\lb\frac{X_i - x^*}{h_n}\rb I(X_i \ge x^*) \cdot \frac{1}{h_n}K\lb\frac{X_j - x^*}{h_n}\rb I(X_j < x^*) \cdot I(Y_j\le Y_i)\right]\\
& = \E \left[\frac{1}{h_n}K\lb\frac{X_i - x^*}{h_n}\rb I(X_i \ge x^*) \cdot \frac{1}{h_n}K\lb\frac{X_j - x^*}{h_n}\rb I(X_j < x^*) \cdot \P(Y_j\le Y_i\mid X_j, X_i)\right]\\
& = \E \left[\frac{1}{h_n}K\lb\frac{X_i - x^*}{h_n}\rb I(X_i \ge x^*) \cdot \frac{1}{h_n}K\lb\frac{X_j - x^*}{h_n}\rb I(X_j < x^*) \cdot \P(Y_j(0)\le Y_i(1)\mid X_j, X_i)\right]\\
& = \E \Bigg[\frac{1}{h_n}K\lb\frac{X_i - x^*}{h_n}\rb I(X_i \in [x^*, x^*+h_n]) \cdot \frac{1}{h_n}K\lb\frac{X_j - x^*}{h_n}\rb I(X_j \in [x^{*}-h_n, x^*)) \\
& \qquad \cdot \P(Y_j(0)\le Y_i(1)\mid X_j, X_i)\Bigg].
\end{align*}
Since $X_i, X_j\rightarrow x^{*}$ as $n\rightarrow \infty$, by Assumption \ref{as:RDD}, $(Y_j, Y_i)\mid (X_j, X_i)$ weakly converges to $(Y_j, Y_i)\mid X_j = X_i = x^{*}$. Thus, for any $\eps > 0$, there exists $n(\eps) > 0$ such that for any $n > n(\eps)$, 
\[\bigg|\P(Y_j(0)\le Y_i(1)\mid X_j, X_i) - \P(Y_j(0)\le Y_i(1)\mid X_j=X_i = x^{*})\bigg|\le \eps.\]
By definition, 
\[\P(Y_j(0)\le Y_i(1)\mid X_j=X_i = x^{*}) = \rate(F_{Y(1)\mid x^*}, F_{Y(0)\mid x^{*}}).\]
Following the same steps as in Lemma \ref{lem:rank_rdd}, 
\begin{align*}
&\E \Bigg[\frac{1}{h_n}K\lb\frac{X_i - x^*}{h_n}\rb I(X_i \in [x^*, x^*+h_n]) \cdot \frac{1}{h_n}K\lb\frac{X_j - x^*}{h_n}\rb I(X_j \in [x^{*}-h_n, x^*))\Bigg]\\
& = \E \Bigg[\frac{1}{h_n}K\lb\frac{X_i - x^*}{h_n}\rb I(X_i \in [x^*, x^*+h_n])\Bigg] \cdot \E\Bigg[\frac{1}{h_n}K\lb\frac{X_j - x^*}{h_n}\rb I(X_j \in [x^{*}-h_n, x^*))\Bigg]\\
& = p^2(x^*)\lb\int_{0}^{1}K(u)du\rb \lb\int_{-1}^{0}K(u)du\rb + o_\P(1).
\end{align*}
Thus, we have proved that 
\[\E\left[\frac{n}{n-1}N\right] = p^2(x^{*}) \lb \int_{0}^{1}K(u)du \rb \lb \int_{-1}^{0}K(u)du \rb \cdot \rate(F_{Y(1)\mid x^*}, F_{Y(0)\mid x^*})+ o_\P(1).\]
Hence, 
\[\E\left[N\right] = p^2(x^{*}) \lb \int_{0}^{1}K(u)du \rb \lb \int_{-1}^{0}K(u)du \rb \cdot \rate(F_{Y(1)\mid x^*}, F_{Y(0)\mid x^*})+ o_\P(1).\]
Next, we bound the variance of $N$. Let $U_{ij}$ be the summand:
\[U_{ij} = \frac{1}{h_n}K\lb\frac{X_i - x^*}{h_n}\rb I(X_i \ge x^*) \cdot \frac{1}{h_n}K\lb\frac{X_j - x^*}{h_n}\rb I(X_j < x^*) \cdot I(Y_j\le Y_i).\]
Clearly, $\Cov(U_{ij}, U_{k\ell}) = 0$ whenever $i,j,k, \ell$ are mutually distinct. For any $i\neq j$,
\begin{align*}
\sqrt{\E[|U_{ij}|]} \le \E\left[\frac{1}{h_n}K\lb\frac{X - x^{*}}{h_n}\rb\right] = O(1),
\end{align*}
where the last equality follows the same steps in the proof of Lemma \ref{lem:rank_rdd}, and 
\begin{align*}
&\sqrt{\Var(U_{ij})}\le \sqrt{\E[U_{ij}^2]} \\
&\le \E\left[\frac{1}{h_n^2}K^2\lb\frac{X - x^{*}}{h_n}\rb\right]\\
& = \frac{1}{h_n}\int_{-1}^{1}K^2(u)p(x^{*} + h_n u)du\\
& = \frac{p(x^*)}{h_n}\int_{-1}^{1}K^2(u)du + O\lb L\int_{-1}^{1}|u|K^2(u)du\rb\\
& = O\lb \frac{1}{h_n}\rb,
\end{align*}
where the second to last line applies Assumption \ref{as:kernel}. Thus, 
\[\E[|U_{ij}|] = O(1), \quad \Var(U_{ij}) = O\lb \frac{1}{h_n^2}\rb.\]
For any mutually distinct $i,j,k$, 
\[\E[|U_{ij}|\cdot |U_{ik}|] \le \E\left[\frac{1}{h_n^2}K^2\lb\frac{X - x^{*}}{h_n}\rb\right]\cdot \E^2\left[\frac{1}{h_n}K\lb\frac{X - x^{*}}{h_n}\rb\right] = O\lb \frac{1}{h_n}\rb.\]
As a result, 
\[|\Cov(U_{ij}, U_{ik})| \le \E[|U_{ij}|\cdot |U_{ik}|] + \E^2[U_{ij}] = O\lb\frac{1}{h_n} + 1\rb = O\lb\frac{1}{h_n}\rb.\]
Putting pieces together, we obtain that 
\[\Var(N) = O\lb \frac{1}{n^2}\Var(U_{ij}) + \frac{1}{n}|\Cov(U_{ij}, U_{ik})|\rb = O\lb \frac{1}{(nh_n)^2} + \frac{1}{nh_n}\rb = o(1).\]
By Chebyshev inequality, 
\[N = \E[N] + O_\P(\sqrt{\Var(N)}) = \E[N] + o_\P(1).\]
Thus, \eqref{eq:goal_mrdd} is proved. 
\end{proof}
\end{document}